\documentclass[twocolumn]{revtex4-2}           
\usepackage{amsmath,mathtools,amsthm}
\usepackage{amssymb}
\usepackage[mathscr]{eucal}
\usepackage{bm,booktabs}
\usepackage{bbm}
\usepackage[shortlabels]{enumitem}
\usepackage{hyperref}
\usepackage[usenames,dvipsnames]{xcolor}
\usepackage{tikz,pgfplots}
\pgfplotsset{compat = newest}

\usepackage{tikz}
\usepackage{graphicx}
\usepackage[caption=false]{subfig}
\usepackage{todonotes}

\newtheorem{theorem}{Theorem}[section]

\newcommand{\pprob}{\mathbb{P}}
\newcommand{\Prob}[1]{\pprob\left(#1\right)}

\newcommand{\expec}{\mathbb{E}}
\newcommand{\Exp}[1]{\expec\left[#1\right]}
\newcommand{\Expp}[1]{\expec_\pprob\left[#1\right]}

\newcommand{\supp}{\textup{supp}}

\begin{document}

\title{Sharpest possible clustering bounds using robust\\ random graph analysis}

\author{Judith Brugman, Johan S.H. van Leeuwaarden and Clara Stegehuis}

\begin{abstract} 
Complex network theory crucially depends on the assumptions made about the degree distribution, while fitting degree distributions to network data is challenging, in particular for scale-free networks with power-law degrees. We present a robust assessment of complex networks that does not depend on the entire degree distribution, but only on its mean, range and dispersion: summary statistics that are easy to obtain for most real-world networks. By solving several semi-infinite linear programs, we obtain tight (the sharpest possible) bounds for correlation and clustering measures, for all networks with degree distributions that share the same summary statistics. We identify various \emph{extremal random graphs} that attain these tight bounds as the graphs with specific three-point degree distributions. We leverage the tight bounds to obtain robust laws that explain how degree-degree correlations and local clustering evolve as function of node degrees and network size. These robust laws indicate that power-law networks with diverging variance are among the most extreme networks in terms of correlation and clustering, building further theoretical foundation for widely reported scale-free network phenomena such as correlation and clustering decay. 
\end{abstract}

\maketitle

\section{Introduction}

Degree heterogeneity drives many complex network properties, with the spread of a virus over a network as a striking example. In homogeneous networks, when differences between node connectivity are relatively small, classical network theory says an epidemic can arise when the average number of secondary infections caused by a single infected individual, $R$, exceeds one. In scale-free networks with high degree fluctuations, on the other hand, this is not a good predictor, as individuals who are infected early on may be different from the average individual. Indeed, these individuals typically have more contacts so that an epidemic can develop even if $R$ is close to zero. A virus then spreads extremely quickly and can hardly be contained.
Many real-world networks, in fact, often have extremely heterogeneous degrees that can be approximated with power-laws, so that the proportion of nodes having $k$ neighbors scales as $k^{-\tau}$ with exponent $\tau$ between 2 and 3 \cite{jeong2000,vazquez2002,faloutsos1999}. Power-law degrees imply various intriguing scale-free network properties such as the absence of an epidemic threshold for $\tau<3$~\cite{janson2009b, pastor2001}, ultra-small distances~\cite{newman2001} and efficient embedding methods~\cite{blasius2018a}.  


Because of this degree heterogeneity, the analysis of such networks is complex.
Network properties such as the friendship paradox, and more generally the connections between nodes with vastly different degrees, are studied in network theory in the form of so-called degree-degree correlations and clustering. Degree-degree correlations measure correlation between the degrees of two connected nodes,  often captured in terms of $a(k)$, the average degree of a neighbor of a degree-$k$ node. By clustering we mean the creation of triangular connections (triadic closure), quantified in terms of $c(k)$, the probability that two neighbors of a degree-$k$ node are neighbors themselves. In uncorrelated networks the $a(k)$ and $c(k)$ are independent of $k$. However, the majority of real-world networks, and scale-free networks in particular, have $a(k)$ and $c(k)$ functions that decay in $k$, first observed in technological networks such as the Internet~\cite{pastor2001b,ravasz2003}. Figure~\ref{fig:youtube} shows the same fall-off for a social network: YouTube users as vertices, and edges indicating friendships between them~\cite{snap}.


When $a(k)$ decreases in $k$, the network is said to be disassortative, so that high-degree vertices typically connect to low-degree vertices. When $c(k)$ decreases in $k$, this may indicate the presence of hierarchy. A hierarchical topology arises, for example, when the rare high-degree nodes together form a backbone, and the low-degree nodes are located in clusters of low-degree nodes that are connected to one of the high-degree nodes. These core peripheries are found in complex networks created by both humans and nature~\cite{gallagher2021}. This view of a hierarchical network explains both the negative degree-degree correlations, because most low-degree nodes are connected to a single high-degree node, and the clustering fall-off, because the core periphery mainly consists of triadic closures between low-degree communities while high-degree nodes rarely participate in triangles and communities. 

Network features such as decaying degree correlations are broadly studied through random graphs, mathematically tractable models that can generate random samples of a graph in which nodes have i.i.d.~degrees~\cite{bianconi2005,itzkovitz2003subgraphs,hofstad2017b,stegehuis2019b,van2021optimal}. Random graph models take the degree distribution as input.
Conditional on the degree distribution, random graph properties such as average distance and clustering can be characterized and tested against measurements from real-world network data with the same degree distribution. 


Motivated by the wide range of examples of networks with heavy-tailed degrees, the power-law distribution has become a popular choice as an input degree distribution for random graph models. 
Fitting a power law to real-world data, however, is statistically challenging~\cite{clauset2009,broido2018,voitalov2019scale}. For small values, a power law is usually not a good fit. For this reason, lower bounds for the power-laws or additional slowly-varying functions are often introduced, but these form extra functions that need to be fitted as well.
Larger values of the power law also present challenges. Most real-world data sets only follow a power law up to some maximal degree, which is often modeled by an exponential cutoff~\cite{newman2001a,mossa2002,ebel2002}. Real-world networks are finite by definition, while a power law allows infinitely large values. 

An inherent disadvantage of network theory that rests on random graphs is the dependency on precise statistical assumptions about the degree distributions. Network theory should not be overly sensitive to the assumed degree distribution, especially when the assumption is hard to justify statistically. For power laws for instance, the tail exponent $\tau$ implies vastly different network properties. One reason for this is the variance of the degree distribution. When the number of nodes $n$ becomes large, the variance grows to infinity for $\tau<3$, while the variance remains finite for $\tau>3$. This difference in variance growth crucially influences the network structure and its degree-degree correlations~\cite{yao2017,stegehuis2017b}. 

 \begin{figure}[tbp]
 \centering
\subfloat[]
    {
        \centering
        \includegraphics[width=0.48\linewidth]{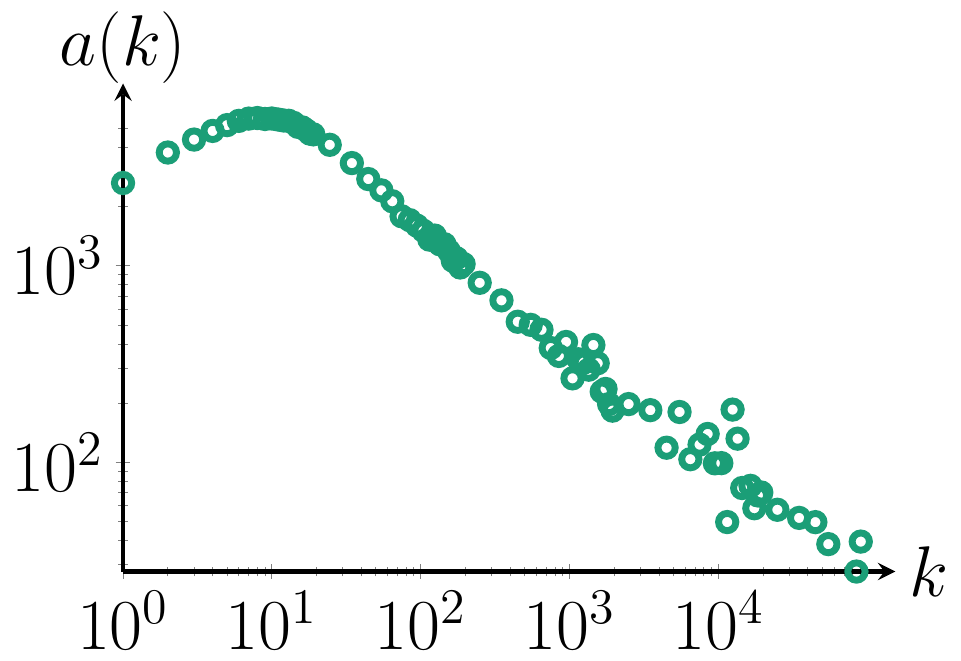}
		\label{fig:youtubeah}
}
\subfloat[]
    {
        \centering
		\includegraphics[width=0.48\linewidth]{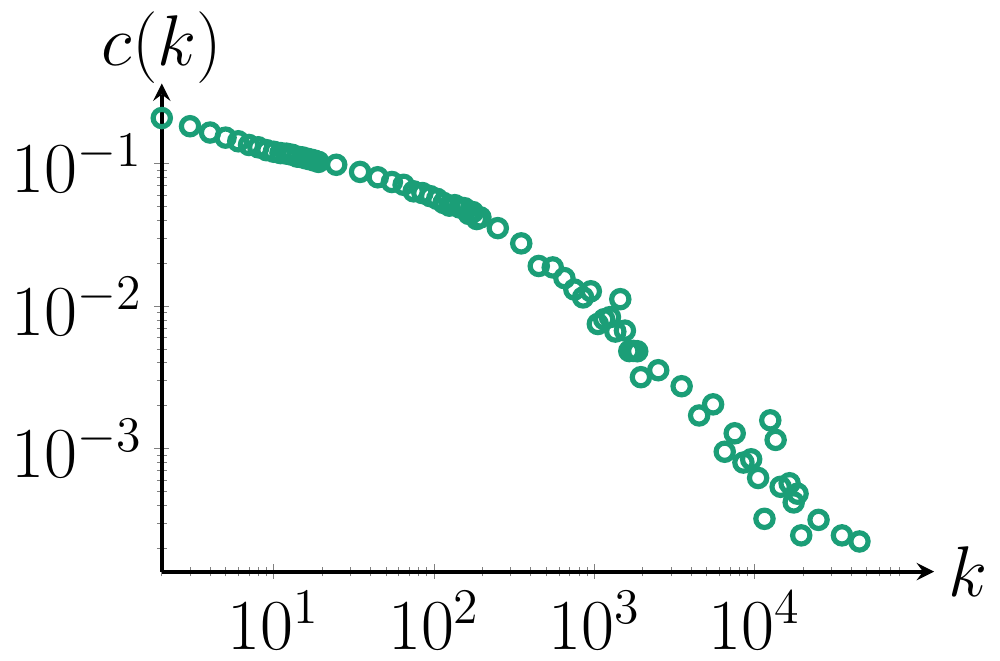}
		\label{fig:youtubech}
}
\caption{a) $a(k)$ and b) $c(k)$ for the YouTube social network}
		\label{fig:youtube}
\end{figure}

To overcome this sensitivity of network null models to precise statistical assumptions on the presence of power-laws or other specific degree distributions, we here characterize degree correlations and clustering in random graphs that only require partial information about the degree distribution. Inspired by the complicated assessment of power laws, we assume only the  mean, dispersion and cutoff of the degree distribution are known. Here we consider two measures of dispersion: the variance and the 
mean absolute deviation (MAD) of the degree distribution. The MAD is an alternative to variance for measuring dispersion around the mean, and may be more appropriate in case of heavy tails. Indeed, MAD can deal with distributions that do not possess a finite variance, in particular the class of power-law distribution with $\tau\in(2,3)$, for which MAD remains finite while variance becomes infinite in the large-network limit when $n\to\infty$.  

We will establish the maximal correlation and maximal clustering that can be achieved by all degree distributions that share the same mean, cutoff and dispersion.
By constructing and solving optimization problems, we find the extremal degree distributions that maximize the degree-degree correlation and clustering. These optimization problems take the form
\begin{equation}\label{genmax}
 \max_{\pprob\in\mathcal{P}}\Expp{{\rm graph \ property}},
\end{equation}
where $\mathcal{P}$ is the set that contains all degree distributions that comply with the limited information, such as the mean, cutoff and dispersion. Hence, within the set $\mathcal{P}$ we find the degree distribution $\pprob^*$ that maximizes the expected graph property. We will refer to the random graph with the extremal degree distribution that attains the maximum as the \emph{extremal random graph}.

We solve optimization problems as in \eqref{genmax}  for the hidden-variable model~\cite{chung2002,boguna2004}, a random graph model that generates graphs with degrees that approximately follow some given distribution. 
The optimization problems in this paper give rise to semi-infinite linear programs and can be solved using methods from distributionally robust optimization. Using a primal-dual approach we can solve these semi-infinite linear programs in closed form, and find the precise description of the degree distribution $\mathbb{P}^*$ that attains the largest expected graph property. Since this distribution is by definition contained in the set $\mathcal{P}$, the bound
\begin{equation}\label{genmaxbound}
\Expp{{\rm graph \ property}}\leq \mathbb{E}_{\mathbb{P}^*}[\rm graph \ property], \quad \forall \mathbb{P}\in \mathcal{P}. 
\end{equation}
is the best possible (tight) bound for all degree distributions that share the same summary statistics as is $\mathcal{P}$.

Distributionally robust optimization finds applications in many domains~\cite{BenTal1972,popescu2005semidefinite,eekelen2019}, but applications in the area of network science are rare. In fact, we are only aware of two papers that apply distributionally robust optimization to study complex networks. The first paper investigates the maximal possible subgraph counts under a restrictive cutoff scheme that creates uncorrelated networks~\cite{leeuwaarden2021}. The second paper provides a distributionally robust model for the influence maximization problem where the influence diffusion is adversarially adapted to the choice of seed set~\cite{chen2020}. Here the authors aim to detect a seed set whose worst-case expected influence is maximized, and show that this differs from the standard model in which influence is assumed to diffuse independently across the different edges. 

We term the largely unexplored approach taken in this paper {\it distributionally robust random graph analysis}, referring to the combination of classical random graph analysis and the optimization framework that only conditions on partial distributional information. This view on random graphs trades precise results that hold for a given degree distribution for robust statements that hold for classes of degree distributions. Such robust results fit well with the search for universal properties of complex networks.

Here are the main contributions of this paper: 
\begin{itemize}
    \item[(i)] For all degree distributions with a given mean, variance and cutoff, we obtain the maximal degree-degree correlations and local clustering. We show that these bounds for $a(k)$ and $c(k)$ decay in $k$, as observed in most real-world networks and random graph models.
    \item[(ii)] We show that the  maximal values of $a(k)$ and $c(k)$ are often attained by uncorrelated graphs. In particular,
    the sharpest possible bounds for $a(k)$ and $c(k)$ for all degree distributions with the same mean and variance are attained by uncorrelated degrees distributions, as long as it is possible to create such uncorrelated distributions with the given mean and variance.
    
    \item[(iii)] We compare the extremal graph models that provide the highest correlations and clustering to existing results for power-law random graphs. While power-laws are often thought of as degree distributions that lead to extreme behavior, the power-laws are not the degree distributions that possess the largest possible values of $a(k)$ and $c(k)$ when $\tau>2$.  
\item[(iv)] We provide a method to detect whether any given real-world data set can be modeled by hidden-variable models for properties of interest. We show that for several real-world data sets, no possible hidden-variable model can model the particular real-world data set.
\end{itemize}

We introduce the hidden-variable model and assumptions on the degree distribution in Section~\ref{sec:model}. We then solve the maximization problem that finds the extremal random graph that generates the maximal degree-degree correlation in Section~\ref{sec:degree}. The scaling laws for clustering as function of the network size are presented in Section~\ref{sec:ck}, and in Section~\ref{sec:N} we do this for clique counts. In Section~\ref{sec:mad} we obtain results for the setting when the MAD instead of the variance is used as dispersion measure. In Section~\ref{sec:powerlaw} we compare the robust bounds obtained in earlier sections
with existing results for scale-free networks with power-law degrees, and with data from real-world networks.

\section{Random graphs and hidden (random) variables}\label{sec:model}

Our analysis of degree-degree correlations and clustering will be based on the hidden-variable model, a random graph model in which every vertex $i\in[n]$ has a weight $h_i$,
and edges are formed between pairs of nodes with a probability that depends on both weights. More specifically, every pair of vertices is connected independently with probability
\begin{equation}\label{eq:pconpf}
	p(h_i,h_j)=\min\Big(\frac{h_ih_j}{h_s^2},1\Big)=\min\Big(\frac{h_ih_j}{{\mu n}},1\Big),
\end{equation} 
where $\mu$ denotes the average weight, and $h_s$ is the structural cutoff set to $\sqrt{\mu n}$ throughout this paper, in line with its typical choice for power-law networks~\cite{boguna2004,colomer2012,catanzaro2005,boguna2003}. This choice of cutoff ensures that the weight of a vertex is close to its degree~\cite{stegehuis2017,boguna2003}. The structural cutoff describes the maximal degree of vertices that are not prone to degree-degree correlations~\cite{boguna2004}. As soon as the degree of a vertex becomes larger than the structural cutoff, it is forced to connect to lower degree vertices, as only few high degree vertices can be present while keeping the average degree fixed.

The natural cutoff describes the constraint on the largest possible network degree, or the largest possible weight, $h_c$. 
In many real-world networks as well as in networks generated from power-law degrees, the largest observed degree is  much larger than the structural cutoff of $\sqrt{\mu n}$. For example, when the degrees of the vertices follow a power-law distribution with degree-exponent $\tau$, the largest degree scales as $n^{1/(\tau-1)}$. 
This means that the network contains vertices that are prone to degree-degree correlations and connection probabilities become non-convex, which makes the network analysis technically more challenging.

The hidden-variable model has several properties that make it amenable to analytical analysis. First of all, when the connection probabilities are chosen suitably, the weight $h$ and the degree $k$ of a vertex are similar with high probability. Indeed, the expected degree $d_j$ of vertex $j$ given its weight, satisfies
\begin{align}
    \Exp{d_j\mid h_j}&=\sum_{j\neq i}\min\Big(\frac{h_ih_j}{{\mu n}},1\Big)\nonumber\\
    & \approx \sum_{i\neq j}\frac{h_ih_j}{{\mu n}} \approx h_j.
\end{align}
To be more precise, when $h\gg 1$, then $k=h(1+o(1))$~\cite{hofstad2017b}. This makes it possible to interchange weights and degrees, which is convenient as the connection probabilities are defined in terms of weights. Secondly, when all hidden variables are assigned, most network statistics of interest can be computed as a function of the hidden variables. For example, the average degree of all neighbors of a vertex with weight $h$ can be written as
    \begin{equation}\label{eq:ak}
    a(h)=\frac{1}{h}\sum_{i=1}^nh_i\min\Big(\frac{h h_i}{\mu n},1\Big),
\end{equation}
where the sum is over all vertices in the network, and multiplies the weight of a vertex with the probability that vertex $i$ connects to the weight-$h$ vertex.

The local clustering coefficient denotes the probability that two randomly chosen neighbors of a vertex with weight $h$ are neighbors themselves. This statistic can again be written as a function of the hidden variables. Formally,
\begin{align}\label{eq:ck}
    &c(h)= \nonumber\\
    & \frac{1}{h^2}\sum_{1\leq i<j\leq n}\min\Big(\frac{h h_i}{\mu n},1\Big)\min\Big(\frac{h h_j}{\mu n},1\Big)\min\Big(\frac{h_i h_j}{\mu n},1\Big).
\end{align}
Here the sum is over all pairs of vertices in the network, and the term inside the summation computes the probability that these vertices form a triangle with the weight-$h$ vertex.

At first sight, degree correlations and clustering and seem unrelated, as the former is defined in terms of edges and the latter in terms of triplets of nodes. Still, intuitively $a(k)$ and $c(k)$ are related in the case of hidden-variable models. Indeed, the average neighbor of a vertex of a weight $k$ vertex is $a(k)$. The probability that two such vertices connect scales as $a(k)^2/n$, when $a(k)$ is sufficiently small. This probability can be interpreted as the probability that two `average' neighbors of a weight-$k$ vertex connect. It turns out that this intuitive reasoning provides the correct scaling of $c(k)$ in some cases. That is, $c(k)\sim a(k)^2/n$~\cite{stegehuis2017b}. 

When studying real-world data sets, we can only observe $\bar c(k)$ and $\bar a(k)$, the local clustering coefficient and average degree of the neigbors of a degree-$k$ vertex, rather than a weight-$h$ vertex. Still, the property of the hidden-variable model that degrees and weights are close makes the difference between these two statistics small in the large-network limit~\cite{hofstad2017b}.

In traditional hidden-variable models, the weights $h_1,\ldots,h_n$ are assumed independent and following some distribution $\mathbb{P}$. The natural cutoff can then be calculated from the distribution $\mathbb{P}$.
In this paper, however, we only specify partial information about the weight (i.e.~degree) distribution. We will assume that for the weights we know the minimal value, their maximal value (the natural cutoff $h_s$), the mean and the dispersion, first measured in variance and later in terms of mean absolute deviation (MAD). Let $h$ denote a generic weight. 
We first assume that the weights are 
sampled independently from a distribution such that (i) $h=h_i$  has support $\supp(h)=[a,h_c]$ with $-\infty<a\leq h_c<\infty$,  (ii) $\Exp{h}=\mu$ and (iii) $\Exp{(h-\mu )^2}=\sigma^2$. This defines the ambiguity set 
\begin{align}\label{eq:psupp}
	\mathcal{P}(\mu,\sigma^2)=& 
	\{\pprob:  \supp(h)\subseteq [a,h_c], \Exp{h}=\mu ,\nonumber\\
	& \quad \Exp{(h-\mu)^2}=\sigma^2\}.
\end{align}
Hence, when we now analyze the hidden-variable model under the assumption that the weight distribution belongs to $\mathcal{P}(\mu,\sigma^2)$, we perform a distributionally robust analysis of the random graph model.

The variance of the degree distribution is often highly affected by the choice of the natural cutoff. In power-law random graphs for example, the variance $\sigma^2$ grows as $h_c^{3-\tau}$. Indeed, for $\tau\in(2,3)$, the variance of the weights can be computed as
\begin{equation}\label{eq:plsigma}
    \int_{1}^{h_c}x^{2-\tau} dx-\mu^2 \sim h_c^{3-\tau}.
\end{equation}
The MAD on the other hand, always satisfies the inequality $d<2\mu$. Thus, as long as the average degree is finite, the MAD will not grow as a function of $h_c.$

In the rest of this paper, we will focus on the graph properties mentioned above, and first seek for the weight distribution $\pprob$ that solves
\begin{equation}\label{genmaxgen}
 \max_{\pprob\in\mathcal{P}(\mu,\sigma^2)}\Expp{\text{graph property}}
\end{equation}
with $\mathcal{P}(\mu,\sigma^2)$ as in~\eqref{eq:psupp}.
This means that we take a distributionally robust approach for the input weights of the hidden-variable model under the assumption that their distribution belongs to $\mathcal{P}(\mu,\sigma^2)$. 
When the graph property~\eqref{genmaxgen} can be viewed as a convex function of the generic weight random variable $h$,~\eqref{genmaxgen} is optimized for a specific distribution with support on three points~\cite{leeuwaarden2021}. Indeed, due to the convex nature of the function, an optimizer aims to put as much weights on the extremal points $a$ and $h_c$, while still adhering to the constraints on the average weight and its variance. This leads to a specific three-point distribution with probability mass on $a, h_c$ and $\mu$.  However, in the setting we now consider with natural cutoff $h_c>\sqrt{\mu n}$, the connection probability~\eqref{eq:pconpf} is not convex, and therefore most graph properties will also not be convex in the hidden variables. In the next sections, we therefore apply a primal-dual based approach to find the distributionally robust graph properties.


\section{Robust degree-degree correlations bounds}\label{sec:degree}
The definition of $a(h)$ in \eqref{eq:ak} assumes that the hidden variables are known. Instead, we now assume that all hidden variables  are drawn from some probability distribution $\pprob$, so that the expected value of $a(h)$ can be computed as
\begin{equation}\label{eq:ahexp}
    \Expp{a(h)}=\frac{n}{h}\Expp{h'\min\Big(\frac{h h'}{\mu n},1\Big)}.
\end{equation}
We then search for the weight distribution $\pprob$ that solves
\begin{equation}\label{genmax2}
 \max_{\pprob\in\mathcal{P}(\mu,\sigma^2)}\Expp{a(h)}
\end{equation}
with $\mathcal{P}(\mu,\sigma^2)$
as in \eqref{eq:psupp}. 
Hence, when we now analyze the hidden-variable model under the assumption that the weight distribution belongs to $\mathcal{P}(\mu,\sigma^2)$, we perform a robust analysis for all distributions with a given mean, variance and cutoff. The optimization problem \eqref{genmax2} can be written as 
\begin{equation}\label{eq:varoptdual}
\begin{aligned}
&\max_{\pprob(x)\geq0} &  &\int_x g(x){\rm d} \pprob(x)\\
&\text{s.t.} &      & \int_x x^2{\rm d}\pprob(x)=\mu^2+\sigma^2, \  \int_x x{\rm d}\pprob(x)=\mu, \\  & & & \int_x {\rm d}\pprob(x)=1,   
\end{aligned}
\end{equation}
where $g(x)=x\min(hx/(\mu n),1)$. 
In optimization theory, \eqref{eq:varoptdual} is called a semi-infinite linear optimization problem (LP).
The Richter-Rogosinski Theorem (see, e.g.,~\cite{rogosinski1958moments,shapiro2009lectures,han2015convex}) says there exists an extremal distribution for problem \eqref{eq:varoptdual} with at most three support points. While finding these points in closed form is typically not possible for general semi-infinite problems, we next show that this is possible for the problem at hand by resorting to the dual problem; see e.g. \cite{Isii1962} and \cite{popescu2005semidefinite}. This dual problem of \eqref{eq:varoptdual} is given by
\begin{equation}\label{eq:varopt}
\begin{aligned}
&\min_{\lambda_1,\lambda_2,\lambda_3} &  &\lambda_1(\mu^2+\sigma^2)+\lambda_2\mu+\lambda_3\\
&\text{s.t.} &      & g(x)-\lambda_1x^2-\lambda_2x-\lambda_3\leq 0 \quad \forall x\in[a,h_c],
\end{aligned}
\end{equation}
and aims to find a tightest quadratic majorant of $g(x)$ that minimizes $\lambda_1(\mu^2+\sigma^2)+\lambda_2\mu+\lambda_3$.
Now $g(x)$ has a quadratic part up to $\min(\mu n/h,h_c)$, and a linear part. Figure~\ref{fig:majorantak} shows that this function has two possible tightest quadratic majorants. The first, $F_1(x)$, is given by $\lambda_1=h/(\mu n), \lambda_2=\lambda_3=0$ and has objective value $h(\mu^2+\sigma^2)/(\mu n)$. The second one, $F_2(x)$, is given by $\lambda_2=1, \lambda_1=\lambda_3=0$, and has objective value $\mu$. Which of the two majorants has the smallest objective value depends on $h$. 
 For low values of $h$, the first majorant gives the lowest objective value, whereas for high values of $h$, the linear one dominates. The changing point is when
\begin{equation}
   h=\mu^2n/(\mu^2+\sigma^2).
\end{equation}

The next step is to find a feasible solution for the primal problem that yields the same objective value as the solution to the dual problem. By weak duality of semi-infinite linear programming, a feasible solution to the dual problem gives a valid upper bound for the optimal primal solution value. A feasible primal solution with an objective value equal to this upper bound would show that strong duality holds. Next, we provide a constructive approach, based on the condition of complementary slackness, to find such a primal solution. 

Assume that strong duality holds. The primal and dual objective are then equal for the primal maximizer $\mathbb{P}^*$ and the dual minimizer $(\lambda_1^*,\lambda_2^*,\lambda_3^*)$, and we can substitute the primal constraints in the dual objective. Hence, we obtain the relation
\begin{equation}\label{eq: compslack}
    \int_x g(x){\rm d} \pprob^*(x) = \int_x \lambda_1^* x^2+\lambda_2^*x+\lambda_3^*\,{\rm d}\mathbb{P}^*(x).
\end{equation}
Since $\lambda_1(\mu^2+\sigma^2)+\lambda_2\mu+\lambda_3 - g(x)\geq0$ pointwise by weak duality, \eqref{eq: compslack} implies that $\mathbb{P}^*$ is supported only on the points where $\lambda_1^* x^2+\lambda_2^*x+\lambda_3^*$ coincides with $g(x)$. 
We now show that in both cases, a three-point distribution achieves the dual objective value.

In the first case, $h<\mu^2n/(\mu^2+\sigma^2)$, and the dual objective value is given by $h(\mu^2+\sigma^2)/(\mu n)$.
Now consider the three-point distribution (for ease of notation, we assume that $a=0$, and denote $l=\min(\mu n/h,h_c)$) on the points $0,\mu,l$, so indeed the quadratic majorant and $g(x)$ coincide. The probabilities are chosen such that the moment conditions are satisfied. We obtain
\begin{equation}\label{eq:3ptcase1}
    p_0=\frac{\sigma^2}{l\mu},\quad p_{\mu}=1-\frac{\sigma^2}{l\mu}-\frac{\sigma^2}{l(l-\mu)},\quad p_l=\frac{\sigma^2}{l(l-\mu)}.
\end{equation}

\begin{figure}
    \centering
    \includegraphics[width=0.4\textwidth]{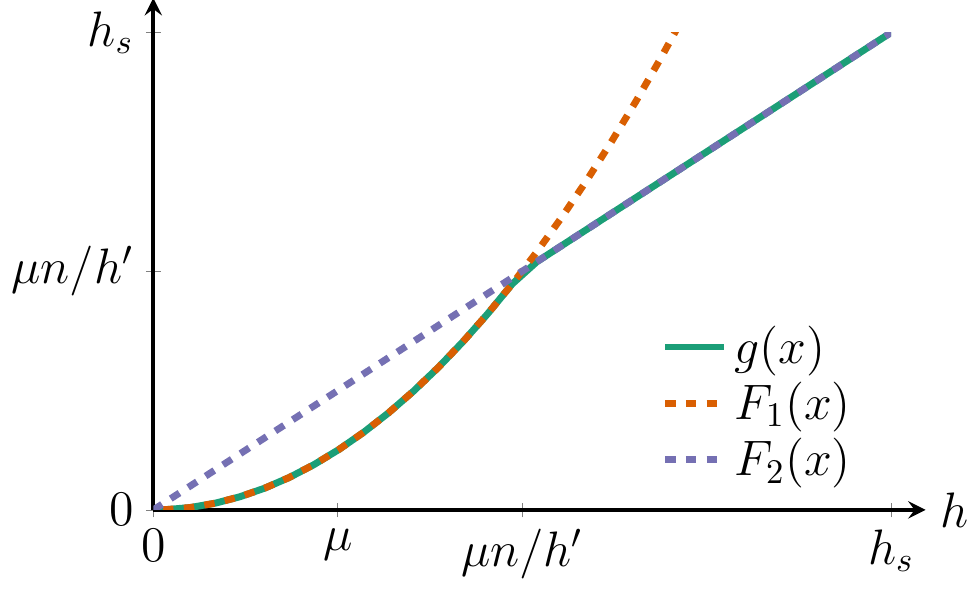}
    \caption{The two possible tightest majorants for the function $g(x)$}
    \label{fig:majorantak}
\end{figure}

This is only a proper probability distribution when $\sigma^2/(l\mu)+\sigma^2/(l(l-\mu))\leq1$, which can be rewritten as  $\sigma^2+\mu^2 \leq \mu l$. We first assume that $\min(\mu n/h,h_c)=\mu n/h$, so we need to check that $\sigma^2+\mu^2 \leq \mu^2n/h$. This follows directly from the assumption on $h$. When $\min(\mu n/h, h_c)=h_c$, we should satisfy $\sigma^2+\mu^2 \leq \mu h_c$. This is always the case, as the maximal variance of a primal solution with mean $\mu$ is given by the primal solution $1-p_0=p_{h_c}=\mu/h_c$, giving as variance $\sigma^2\leq \mu h_c-\mu^2$.

This three-point distribution gives as objective value for the primal problem
\begin{equation}
    \left(1-\frac{\sigma^2}{l\mu}-\frac{\sigma^2}{l(l-\mu)}\right)\frac{\mu^2h}{\mu n}+\frac{\sigma^2}{l(l-\mu)}\frac{l^2h}{\mu n}=\frac{(\sigma^2+\mu^2)h}{\mu n},
\end{equation}
which is equal to the dual objective value.

Thus, by duality, this three-point degree distribution generates the extremal random graph for $a(k)$,  and the result that for  $h<\mu^2n/(\mu^2+\sigma^2)$,
\begin{equation}
        \max_{\pprob\in\mathcal{P}^*(\mu,\sigma^2)}\Expp{a(h)}=\frac{\mu^2+\sigma^2}{\mu}.
    \end{equation}

In the second case, $h>\mu^2n/(\mu^2+\sigma^2)$, and the dual objective value is given by $\mu$. Here, consider the three-point distribution
\begin{align}\label{3pointA}
    p_0& =1-p_{\mu n/h}-p_{h_c},\quad  
    p_{\mu n/h}=\frac{h^2 \left(-h_c \mu +\mu ^2+\sigma ^2\right)}{\mu  n (\mu  n-h_c h)},
    \nonumber\\
    p_{h_c}&=\frac{\mu ^2 (h-n)+h \sigma ^2}{h_c (h_c h-\mu  n)}.
\end{align}
This is only a proper distribution when $\sigma^2<h_c\mu-\mu^2$, which always holds by a similar reasoning as for the first case. The second condition of $\sigma^2>h_c\mu-\mu^2-\mu n/h^2(h_ch-\mu n)$ is more involved, and is in fact not necessary. Indeed, in~\eqref{3pointA}, we chose the third point of the distribution at $h_c$. However, as the tightest majorant in Figure~\ref{fig:majorantak} touches on an entire line, it is also possible to choose the third point somewhere else in $[\mu n/h,h_c]$ while achieving the same optimal value. Choosing another point for the three-point distribution also results in different conditions on $\sigma^2$. The lowest such constraint is when $p_0=1-p_{\mu n/h}$ in~\eqref{3pointA}, yielding $\sigma^2=\mu^2n/h-\mu^2$, which is ensured by our assumption on $h$. 
Under this three-point distribution, $\mathbb{E}[X]=\mu$, $\mathbb{E}[(X-\mu)^2]=\sigma^2$ and the primal objective value $\mathbb{E}[g(X)]=\mu$. 

Thus, by duality, the three-point distribution is the variance-based extremal random graph for $a(k)$, giving that for $h>\mu^2n/(\mu^2+\sigma^2)$,
\begin{equation}
        \max_{\pprob\in\mathcal{P}(\mu,\sigma^2)}\Expp{a(h)}=\frac{\mu n}{h}.
    \end{equation}
 Notice that the three-point distribution \eqref{3pointA} is not a unique optimum, as the dual function $F_2(x)$ coincides with $g(x)$ on the entire interval $[\mu n/h,h_c]$. Therefore, one can construct an arbitrary (discrete, continuous or mixed) probability distribution with support on the interval $[\mu n/h,h_c]$, as long as the mean and variance conditions are satisfied. Similarly, the three-point distribution \eqref{eq:3ptcase1} is also not unique.

This yields the following theorem:
\begin{theorem}\label{thm:anndvariance}
    When $p(h,h')=\min(hh'/(\mu n),1)$, and $h_c\ll n$,
    \begin{equation}
        \max_{\pprob\in\mathcal{P}(\mu,\sigma^2)}\Expp{a(h)}=\begin{cases}\frac{\mu^2+\sigma^2}{\mu} & h<\mu^2n/(\mu^2+\sigma^2)\\
       \frac{\mu n}{h} &h>\mu^2n/(\mu^2+\sigma^2).
        \end{cases}
    \end{equation}
\end{theorem}
The theorem distinguishes two regimes: one constant regime for low $h$, and a decaying regime for high enough $h$. In the constant regime, vertices are not prone to degree-degree correlations: all vertices have the same average nearest neighbor degree as long as $h<\mu^2n/(\mu^2+\sigma^2)$. Furthermore, higher degree variance implies a lower threshold, and hence more vertices that are subject to degree-degree correlations.
This is consistent with the intuition that 
degree-degree correlations arise because high-degree vertices are forced to connect to lower-degree vertices due to the lack of sufficiently many high-degree vertices. When $h_c<\mu^2n/(\mu^2+\sigma^2)$, it is possible to create entirely uncorrelated networks. This condition is more general than the often used $h_c\ll\sqrt{n}$ constraint for uncorrelated networks that was found for power-law networks~\cite{boguna2004}, as here we do not assume any specific weight distribution.

\section{Robust clustering bounds}\label{sec:ck}
We now consider the probability that two randomly chosen neighbors of a vertex of weight $h$ are connected to one another as well, $c(h)$. Again, in~\eqref{eq:ck} the hidden variables are assumed to be fixed. We assume that the hidden variables are drawn from some distribution $\pprob$, so that the expected value of $c(h)$ becomes
\begin{align}\label{eq:expck}
    &\Expp{c(h)}=\nonumber\\
    &\frac{n^2}{h^2}\Expp{\min\Big(\frac{h h_1}{\mu n},1\Big)\min\Big(\frac{h h_2}{\mu n},1\Big)\min\Big(\frac{h_1 h_2}{\mu n},1\Big)}
\end{align}
Now, instead of optimizing over only one hidden-variable as in~\eqref{eq:varoptdual}, we need to jointly optimize the weight distributions of both neighbors. Under variance-based optimization, the optimization problem for $c(h)$ corresponding to~\eqref{eq:varoptdual} becomes
\begin{equation}\label{eq:varoptdualch}
\begin{aligned}
&\max_{\pprob(x)\geq0} &  &\int_{x_1}\int_{x_2} g(x_1,x_2){\rm d} \pprob(x_2){\rm d} \pprob(x_1)\\
&\text{s.t.} &      & \int_x x^2{\rm d}\pprob(x)=\mu^2+\sigma^2, \int_x x{\rm d}\pprob(x)=\mu,\\
&  & & \int_x {\rm d}\pprob(x)=1,  
\end{aligned}
\end{equation}
where 
\begin{equation}
g(x_1,x_2)=\min\Big(\frac{x_1x_2}{\mu n},1\Big)\min\Big(\frac{x_1h}{\mu n},1\Big)\min\Big(\frac{x_2h}{\mu n},1\Big).\end{equation} 

It turns out that this optimization problem is difficult to solve, due to the constraint that the two variables $x_1$ and $x_2$ are i.i.d.. We therefore instead solve a relaxation of~\eqref{eq:varoptdualch}, where we allow these two variables to be correlated, or drawn from different distributions. This relaxed problem takes the form 
\begin{equation}\label{eq:varoptrelaxed}
\begin{aligned}
&\max_{\pprob(x_1,x_2)\geq0} &  &\int_{x_1}\int_{x_2} g(x_1,x_2){\rm d} \pprob(x_1,x_2)\\
&\text{s.t.} &      & \int_{x_1}\int_{x_2} x_1^2x_2^2{\rm d}\pprob(x_1,x_2)=(\mu^2+\sigma^2)^2,\\
& & & \int_{x_1}\int_{x_2} x_1x_2{\rm d}\pprob(x_1,x_2)=\mu^2,\\
&  & & \int_{x_1}\int_{x_2} {\rm d}\pprob(x_1,x_2)=1,   
\end{aligned}
\end{equation}
which again is a semi-infinite linear optimization problem. 
Here, instead of optimizing over a single distribution from which both weights are drawn, we optimize over a joint, symmetric distribution over the weights of the other two vertices involved in a triangle, $\pprob(x_1,x_2)$. As a consequence, it is possible that the optimal weight distributions found by solving the relaxed problem are correlated 
so that the two vertices jointly optimize their weights to make a triangle formation more likely. Interestingly, it turns out that this is not the case. In Appendix~\ref{sec:ckvarbasedproof}, we show that the optimizer of the relaxed optimization problem, that thus allows for correlations and is easier to solve, is in fact an uncorrelated distribution, and therefore also the solution of the original optimization problem~\eqref{eq:varoptdualch}. We are able to derive this optimizer, because the dual version of this problem
\begin{equation}\label{eq:ckdual}
\begin{aligned}
&\min_{\lambda_1,\lambda_2,\lambda_3} &  &\lambda_1(\mu^2+\sigma^2)^2+\lambda_2\mu^2+\lambda_3\\
&\text{s.t.} &      & g(x_1,x_2)-\lambda_1x_1^2x_2^2-\lambda_2x_1x_2-\lambda_3\leq 0 \nonumber\\
& & & \quad \forall x_1,x_2\in[a,h_c],
\end{aligned}
\end{equation}
is relatively easy to solve.
This gives the following theorem on the distributionally robust optimizer for $c(h)$:

\begin{theorem}\label{thm:ckvariance}
    When $p(h,h')=\min(hh'/(\mu n),1)$ and $\sigma^2<\mu \max(\sqrt{\mu n},\mu n/h)-\mu^2$,
    \begin{align}
        &\max_{\pprob\in\mathcal{P}(\mu,\sigma^2)}\Expp{c(h)}\nonumber\\
        & =\begin{cases}\min\Big(\frac{(\mu^2+\sigma^2)^2}{\mu^3n},1\Big) & h<\mu^2n/(\mu^2+\sigma^2)\\
       \frac{\mu n}{h^2} &h>\mu^2n/(\mu^2+\sigma^2).
        \end{cases}
    \end{align}
\end{theorem}
This theorem only holds when $\sigma^2$ is not too large. We conjecture that when $\sigma^2$ is larger than the range prescribed by the theorem, a primal-dual gap of~\eqref{eq:varoptrelaxed} appears, indicating that the optimization problem cannot be solved anymore through primal-dual methods. Indeed, for larger $\sigma^2$, the best dual solution remains feasible as it does not depend on $\sigma^2$. However, it is then impossible to construct a probability distribution with the required variance on the set of values where the dual constraints are tight. This suggests that a primal-dual gap is present in that case, as there is no primal feasible solution that satisfies complementary slackness.
For large $\sigma^2$, this implies that the primal problem has to be solved without the help of the dual problem, which makes the problem significantly more challenging. 

\section{Robust clique counts}\label{sec:N}

Whereas $a(k)$ and $c(k)$ 
measure two- and three-point correlations between nodes, we demonstrate in this section that robust bounds can also be obtained for network statistics that include more than three nodes. We focus on the number of cliques of size $k$, denoted as  $N(K_k)$, and use that
 \begin{align}
     \Exp{N(K_k)}& =\sum_{1\leq i_1<i_2\dots <i_k}\prod_{u<v}p(h_{i_u},h_{i_v})\nonumber\\
     & =\sum_{1\leq i_1<i_2\dots <i_k}\prod_{u<v}\min\Big(\frac{h_{i_u}h_{i_v}}{\mu n},1\Big),
 \end{align}
 as a clique is present if and only if all possible edges between nodes $i_1,i_2\dots ,i_k$ are present.
To establish the variance-based bound on the expected number of cliques, 
we formulate the multivariate optimization problem
\begin{equation}\label{eq:varoptdualclique}
\begin{aligned}
&\max_{\pprob(x)\geq0} &  &\int_{x_1}\dots\int_{x_k} g(x_1,x_2,\dots,x_k){\rm d} \pprob(x_1)\dots{\rm d} \pprob(x_k)\\
&\text{s.t.} &      & \int_x x^2{\rm d}\pprob(x)=\mu^2+\sigma^2, \int_x x{\rm d}\pprob(x)=\mu,\\
&  & & \int_x {\rm d}\pprob(x)=1  
\end{aligned}
\end{equation}
with 
$$g(x_1,\dots,x_k)=\prod_{1\leq i<j\leq k}\min\Big(\frac{x_ix_j}{\mu n},1\Big).$$
As for clustering, this optimization problem appears intractable due to the i.i.d.~hidden variables. 
We therefore again solve a relaxation of~\eqref{eq:varoptdualclique} instead that drops the i.i.d.~assumption. 
In Appendix~\ref{sec:ckvarbasedproof} we solve this relaxed optimization problem and prove the following robust bounds for clique counts:
\begin{theorem}\label{thm:Kkvar}
When $\sigma^2\leq \mu(\sqrt{\mu n}-\mu)$, $k> 3$, and as $n\to\infty$,
\begin{equation}\label{eq:Kkoptvarnocutoff}
    \max_{\pprob\in\mathcal{P}(\mu,\sigma^2)}\Expp{N(K_k)}
    =\frac{(\mu^2+\sigma^2)^k}{\mu^{k}k!}(1+o(1)).
\end{equation}
Furthermore, for $k=3$, $\sigma^2\leq \mu(\sqrt{\mu n}-\mu)$ and for any $n$,
\begin{equation}\label{eq:Kkoptvarnocutoff3}
    \max_{\pprob\in\mathcal{P}(\mu,\sigma^2)}\Expp{N(K_k)}
    =\frac{(\mu^2+\sigma^2)^k}{\mu^{k}k!}.
\end{equation}
\end{theorem}

This theorem shows that cliques significantly increase when the variance grows, as is the case for heavy-tailed weight distributions. 
The theorem only holds asymptotically (except for $k=3$), as we create primal and dual solutions with a small gap between their respective optimal values that vanishes as $n\to\infty$. Whereas Theorems~\ref{thm:anndvariance}and~\ref{thm:ckvariance} gave exact (non-asymptotic) results, here the relaxed optimization method that provided exact results for Theorem~\ref{thm:ckvariance} gives non i.i.d.~weight distributions. Thus, this method does not provide exact bounds on clique counts. Instead, we first solve the dual problem, and then construct i.i.d.~primal weight distributions that asymptotically achieve the dual value, and are therefore asymptotically optimal.

As in Theorem~\ref{thm:ckvariance}, the theorem contains a condition on $\sigma^2$. We conjecture that for larger $\sigma^2$, a larger primal-dual gap is present that does not vanish in the large-network limit, so that the optimization problem~\eqref{eq:varoptdualclique} cannot be solved through its dual variant, similarly as for Theorem~\ref{thm:ckvariance}.

\section{MAD instead of variance}\label{sec:mad}


We now turn to a second measure of dispersion: mean absolute deviation. While the variance of a random variable can be infinite, the MAD is always bounded by $2\mu$, so that even in networks with heavy-tailed degree distributions this quantity remains finite. For maximizing based on MAD, the ambiguity set now becomes
\begin{align}\label{eq:psuppmad}
	&\mathcal{P}(\mu,d)=\nonumber\\
	&	\{\pprob:  \supp(h)\subseteq [a,h_c], \Exp{h}=\mu ,\Exp{|h-\mu|}=d\}.
\end{align}
As for the variance-based approach, we then aim to find the probability distribution $\pprob\in\mathcal{P}(\mu,d)$ that maximizes the network statistics $a(h)$ and $c(h)$. We can use the same approach of constructing an optimization problem based on the constraints formed by the ambiguity set and finding the optimal primal-dual solution. As shown in Appendix \ref{sec:proofsMAD}, we obtain the following result for $a(h)$:

\begin{theorem}\label{thm:annd}
    When $p(h,h')=\min(hh'/(\mu n),1)$, and $h_c\to\infty$ as $n\to\infty$ and $h_c\ll n$,
    \begin{equation}
        \max_{\pprob\in\mathcal{P}(\mu,d)}\Expp{a(h)}=\frac{d}{2\mu}\min\Big(h_c,\frac{\mu n}{h}\Big)(1+o(1)).
    \end{equation}
\end{theorem}

This optimal value of $a(h)$ is attained by the 3-point distribution
\begin{equation}\label{eq:3pointah}
	p_0=\frac{d}{2\mu},\quad p_\mu=1-\frac{d}{2\mu}-\frac{d}{2(l-\mu)},\quad p_{l}=\frac{d}{2(l-\mu)},
\end{equation}
where $l=\min(\mu n/h,h_c)$.

To obtain results for $c(k)$ with MAD as well, we need to solve the optimization problem
\begin{equation}\label{eq:varoptchmad}
\begin{aligned}
&\max_{\pprob(x)\geq0} &  &\int_{x_1}\int_{x_2} g(x_1,x_2){\rm d} \pprob(x_2){\rm d} \pprob(x_1)\\
&\text{s.t.} &      & \int_x |x-\mu|{\rm d}\pprob(x)=d, \int_x x{\rm d}\pprob(x)=\mu,\\
&  & & \int_x {\rm d}\pprob(x)=1,   
\end{aligned}
\end{equation}
similarly to~\eqref{eq:varoptdualch}. Again, this optimization problem is difficult to solve, due to the constraint that the two variables $x_1$ and $x_2$ are i.i.d.. We could solve this problem by proceeding as in Section~\ref{sec:ck}, by writing an unconstrained version of this optimization problem, where we allow the two variables to be non-identical or correlated. However, these techniques then lead to the dual problem
\begin{equation}
\begin{aligned}
&\min_{\lambda_1,\lambda_2,\lambda_3}   \lambda_1d^2+\lambda_2\mu^2+\lambda_3\\
&\text{s.t.}      \    g(x_1,x_2)-\lambda_1|x_1-\mu||x_2-\mu|-\lambda_2x_1x_2-\lambda_3\leq 0\nonumber\\
& \forall x_1,x_2\in[a,h_c].
\end{aligned}
\end{equation}
Compared with~\eqref{eq:ckdual}, this dual function is no longer quadratic, but a product of absolute values summed with a linear term. The dual then tries to find the tightest majorant of this two-dimensional function of the non-convex function $g(x_1,x_2)$ illustrated in Figure~\ref{fig:gfunct}. However, finding the tightest majorant of $\lambda_1|x_1-\mu||x_2-\mu|-\lambda_2x_1x_2-\lambda_3$ to a general function is not obvious, as it is a function with a kink and different behavior near the $x$ and $y$ axes from the line $y=x$, as illustrated in Figure~\ref{fig:madfunct}. This makes it more difficult to find the values of $\lambda_1,\lambda_2,\lambda_3$ that find this tightest majorant. 

Therefore, we take a different approach instead. We first focus on the one-dimensional problem: for fixed $x_1$, what is the distribution of $x_2$ that maximizes $g(x_1,x_2)$? This problem can now be solved by a one-dimensional dual problem, which is easy to solve, similarly as in Section~\ref{sec:degree}. We then take the distribution of the optimal $x_2$ (which may depend on the value of $x_1$), and then optimize over the distribution of $x_1$ as well. Now this iterative approach may introduce correlations between the distributions: it is possible that the optimal distribution for $x_1$ is different from, or dependent on, the optimal distribution of $x_2$. Still, in some cases it may be true that the output of this less restrictive optimization problem gives an i.i.d. distribution of $x_1$ and $x_2$. In that case, this method also solves~\eqref{eq:varoptchmad}. In the case when $g(x_1,x_2)$ is convex, this is known to be true. Unfortunately, for the case of computing $c(h)$, $g(x_1,x_2)$ is not convex. And indeed, the iterative optimizer is not always i.i.d. in this case. Still, we show that asymptotically, an i.i.d. optimizer achieves the same maximal value of $c(h)$. Applying this iterative method, as shown in Appendix \ref{sec:proofsMAD}, leads to the following results:

\begin{theorem}\label{thm:ckhigh}
    When $p(h,h')=\min(hh'/(\mu n),1)$ and $h_c\to\infty$ as $n\to\infty$,
    \begin{align}
       & \max_{\pprob\in\mathcal{P}(\mu,d)}\Expp{c(h)}\nonumber\\
       &=\begin{cases}
        1 & h\ll \sqrt{\mu n}  \enspace\text{and } \\
         & 2\mu(1-\sqrt{\mu/n})<d<2\mu(1-h/n)\\
        \frac{d^2}{4\mu^2 }(1+o(1)) &h\ll \sqrt{\mu n} \enspace \text{and} \enspace d<2\mu(1-\sqrt{\mu/n})\\
        \frac{d^2n}{4\mu h^2}(1+o(1)) & \sqrt{\mu n}\ll h\ll n.
        \end{cases}
    \end{align}
\end{theorem}
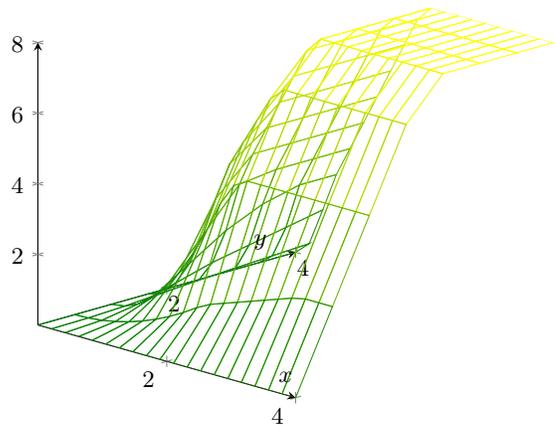
\begin{figure}
\centering
\begin{tikzpicture}
\begin{axis}[
    view={45}{20},
    xlabel=$x$,
    ylabel=$y$,
    axis lines=center, colormap/greenyellow]

\addplot3 [
    domain=0:4,
    domain y = 0:4,
    samples = 20,
    samples y = 8,
    xlabel = {$x$},
    ylabel = {$y$},
    mesh] {min(x,2)*min(y,2)*min(x*y,2)};
 
\end{axis}
\end{tikzpicture}
\caption{The function $g(x,y)$.}
\label{fig:gfunct}
\end{figure}

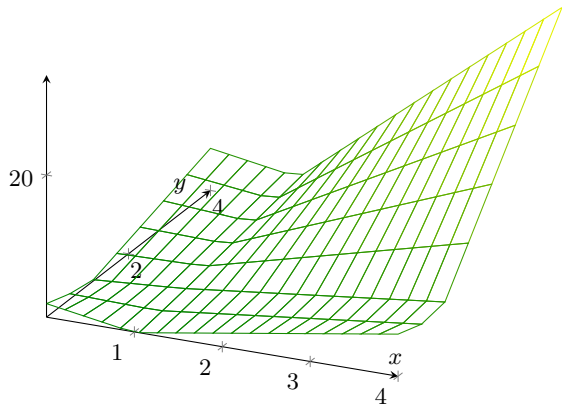
\begin{figure}
\centering
\begin{tikzpicture}
\begin{axis}[
    xlabel=$x$,
    ylabel=$y$,
    axis lines=center, colormap/greenyellow]

\addplot3 [
    domain=0:4,
    domain y = 0:4,
    samples = 20,
    samples y = 8,
    xlabel = {$x$},
    ylabel = {$y$},
    mesh] {2*abs(x-1)*abs(y-1)+x*y};
 
\end{axis}
\end{tikzpicture}
\caption{The function $2|x-1||y-1|+xy$.}
\label{fig:madfunct}
\end{figure}

 While from the perspective of i.i.d.~weight sampling of the random graph it is natural to constrain the two vertices that form a triangle together with the degree-$h$ node to be sampled from the same distribution, in the optimization problem it is also possible to find the optimal pair of correlated distributions over the two nodes that form a triangle together with the weight-$h$ node. In that case, the nodes that form a triangle together with the weight-$h$ node are both sampled from the ambiguity set $\mathcal{P}_{(\mu,d)}$, so that they still both have mean weight $\mu$ and MAD $d$. However, the weights of the two nodes are now allowed to be correlated. 

In the proof of Theorem~\ref{thm:ckhigh}, we show that in the ranges where Theorem~\ref{thm:ckhigh} is valid, adding correlations between the distributions of $h_1$ and $h_2$ does not make a difference in the scaling for $c(h)$. Indeed, the upper bound of 1 is always valid for $c(h)$, so that lifting the constraint on the i.i.d.~distributions of $h_1$ and $h_2$ cannot increase the optimal value of $c(h)$. For the second setting where $\sqrt{\mu n}\ll h\ll n$, we in fact show in the proof of Theorem~\ref{thm:ckhigh} that by optimizing $c(h)$ over a set of probability distributions on $h_1$ and $h_2$ which are allowed to be correlated, we end up with i.i.d. distributions as the optimizer. So here also lifting the constraint on the joint distribution will not increase $c(h)$. Furthermore, in the case $h\ll\sqrt{\mu n}$ and $d$ is small, allowing correlations only increases the non-leading order terms, so that asymptotically, $c(h)$ cannot be increased by allowing correlations between the distributions of $h_1$ and $h_2$.

\section{Comparison with other networks}\label{sec:powerlaw}

We now compare the extremal values of $a(k)$ and $c(k)$ to several synthetic and real network data.
\paragraph{Erd\H{o}s-R\'enyi model.}
In the Erd\H{o}s-R\'enyi model, $h_i=\mu, \forall i$, and all vertices connect with probability $\mu/n$. Thus, $\sigma^2=0$ and $d=0$. The zero variance and MAD means that there is only one distribution in the ambiguity sets~\eqref{eq:psupp} and~\eqref{eq:psuppmad}, which is $h_i=\mu, \forall i$. Therefore, Theorem~\ref{thm:anndvariance} and~\ref{thm:annd} predict that the maximal value of $a(h)=\mu$, which is the exact average weight of a neighbor in an Erd\H{o}s-R\'enyi random graph, as all vertices have weight $\mu$. Furthermore, Theorem~\ref{thm:ckvariance} and~\ref{thm:ckhigh} predict that the maximal value of $c(h)=\mu/n$, which is also equal to $c(h)$ in an Erd\H{o}s R\'enyi model, as the probability that two neighbors of a vertex connect is $\mu/n$, the same as the connection probability for all pairs of vertices. Thus, for Erd\H{o}s-R\'enyi random graphs, our bounds are tight.

\paragraph{Poisson random graph.}
When the hidden-variables have a Poisson distribution with mean $\mu$, the second moment of the weight distribution is $\mu+\mu^2$. 
For such networks, $hh'<\mu n$ almost surely for all $h\ll n$. Thus,~\eqref{eq:ahexp} gives $\Exp{a(h)}=1+\mu$. Applying Theorem~\ref{thm:anndvariance} yields that the maximal possible $a(h)$ value among all networks with the same mean and variance of the weights also equals $1+\mu$. Similarly $\Exp{c(h)}=(\mu +\mu^2)^2/(\mu^3n)$ for Poisson random graphs by Equation~\eqref{eq:expck}, while by Theorem~\ref{thm:ckvariance}, the maximal value of $c(k)$ among all power-law random graphs also equals $(\mu +\mu^2)^2/(\mu^3n)$. Thus, Poisson random graphs achieve the maximal bounds for $a(h)$ and $c(h)$ exactly. 

When looking at the MAD-based bounds, the picture changes drastically. Indeed, for a Poisson random variable with integer mean $\mu$~\cite{ramasubban1958}
\begin{equation}
    d=2\mu^{\mu+1}e^{-\mu}/\mu!,
\end{equation}
so that by Theorem~\ref{thm:annd} the maximal value of $a(h)$ scales for low $h$ as $h_c\mu^{\mu}e^{-\mu}/\mu!,$ which can be much larger than the achieved value of $1+\mu$. 

For $c(k)$, Theorem~\ref{thm:ckhigh} yields that for $h$ low, the maximal value of $c(h)$ for random graphs with the same mean and MAD as the Poisson random graph equals $(\mu^{\mu}e^{-\mu}/\mu! )^2$, which is an $n$-independent constant, in contrast to the achieved value of $c(h)=(\mu +\mu^2)^2/(\mu^3n)$, which decays in $n$. Thus, our variance-based bounds can be significantly lower than the ones based on equal MAD. 

\paragraph{Comparing power-law $a(h)$ to extremal $a(h)$.}
We now turn to random graphs with power-law distributed weights. 
We first compare the maximal scaling of $a(h)$ given by Theorem~\ref{thm:anndvariance} to the value of $a(h)$ attained by the power-law weight distribution
\begin{equation}
    \Prob{h>x}=Cx^{1-\tau}.
\end{equation}
When sampling $n$ i.i.d.~weights from this distribution, the maximal weight scales as $h_c=n^{1/(\tau-1)}$ with high probability.
In such power-law Chung-Lu models~\cite{yao2017,stegehuis2017b},
\begin{align}\label{eq:ahCL}
    a(h)\sim \begin{cases} n^{(3-\tau)/(\tau-1)} & h\ll n^{(\tau-2)/(\tau-1)}\\
     (n/h)^{3-\tau}& h\gg n^{(\tau-2)/(\tau-1)}.
    \end{cases}
\end{align}
We now investigate how close this value of $a(h)$ is to the maximal possible values among all Chung-Lu models with the same mean and variance as the power-law distribution. For power-law distributed weights, $\sigma^2\sim n^{(3-\tau)/(\tau-1)}$, as derived in~\eqref{eq:plsigma}. Thus, Theorem~\ref{thm:anndvariance} yields
\begin{align}\label{eq:aklow}
\max_{\pprob\in\mathcal{P}(\mu,n^{(3-\tau)/(\tau-1)})} \Expp{a(h)} &= \frac{\mu^2+n^{(3-\tau)/(\tau-1)}}{\mu}\nonumber\\
& \sim n^{(3-\tau)/(\tau-1)},
\end{align}
when $h<n^{(2\tau-4)/(\tau-1)}$, while 
\begin{align}\label{eq:akhigh}
\max_{\pprob\in\mathcal{P}(\mu,n^{(3-\tau)/(\tau-1)})} \Expp{a(h)} &= \frac{\mu n}{h},
\end{align}
when $h>n^{(2\tau-4)/(\tau-1)}$
For large $h$, the scaling in~\eqref{eq:ahCL} only agrees with the value of $\mu n/h$ of~\eqref{eq:akhigh} for $\tau=2$. For low $h$, the power-law value of $a(h)$ of~\eqref{eq:ahCL} agrees with the scaling of~\eqref{eq:aklow} for all $\tau\in(2,3)$. This indicates that a power-law distribution asymptotically achieves the most extreme values of $a(h)$ possible for $h$ small among all random graphs with the same mean and variance of the degree distribution. For larger $h$, the extremal $a(h)$ scaling is only attained for power-law random graphs for $\tau=2$. Indeed, Figure~\ref{fig:akplvar} illustrates that the variance-based upper bound on $a(k)$ is close to the value achieved by a power-law Chung-Lu model when $\tau\approx 2$, and that power-law graphs with higher degree-exponents have a closer gap with the maximal possible $a(h)$ value.

We now again compare $a(h)$ of the the power-law distribution~\eqref{eq:ahCL} to a matching extremal value, but now the extremal random graph based on a matching mean and MAD. For power-laws~\cite{leeuwaarden2021},
\begin{equation}\label{eq:pld}
	d=\frac{C(2\mu^{2-\tau}-1-h_c^{2-\tau}))}{\tau-2}+\frac{C\mu(-2\mu^{1-\tau}+1+h_c^{1-\tau}))}{\tau-1}.
\end{equation}
Thus, for $\tau>2$, $d$ is approximately constant.
Comparing Theorem~\ref{thm:annd} where we take $d$ constant and $h_c=n^{1/(\tau-1)}$ with~\eqref{eq:ahCL} shows that 
\begin{align}
\max_{\pprob\in\mathcal{P}^*(\mu,d)} \Expp{a(h)} &=\begin{cases} \frac{d}{2\mu}n^{1/(\tau-1)} & h<\mu n^{(\tau-2)/(\tau-1)}\\
\frac{dn}{2h} & h>\mu n^{(\tau-2)/(\tau-1)}.
\end{cases}
\end{align}
Comparing this with~\eqref{eq:ahCL}, shows that for $\tau=2$, the maximal degree-degree correlations among all Chung-Lu random graphs with given MAD, and $\mu$ scales the same as for the power-law distribution. Still, Figure~\ref{fig:akplmad} shows that the differences in constants in the change point as well as in the maximal scaling make the power-law $a(h)$ to be quite far from the MAD-based bound, even for $\tau\approx 2$. Note also that the MAD-based bounds sometimes drops below the actual power-law value for $h$ large. This is because we plot the values of $a(k)$ all the way up to $h=n$, while the ($\tau$-dependent) cutoff lies already at $h_c=n^{1/(\tau-1)}$, which is close to 400 for $\tau=2.9$ for example. Thus, for all $h$ in $[1,h_c]$, the MAD-bound is a valid upper bound.

\begin{figure*}[tbp]
    \centering
    \subfloat[]
    {
    \includegraphics[width=0.4\textwidth]{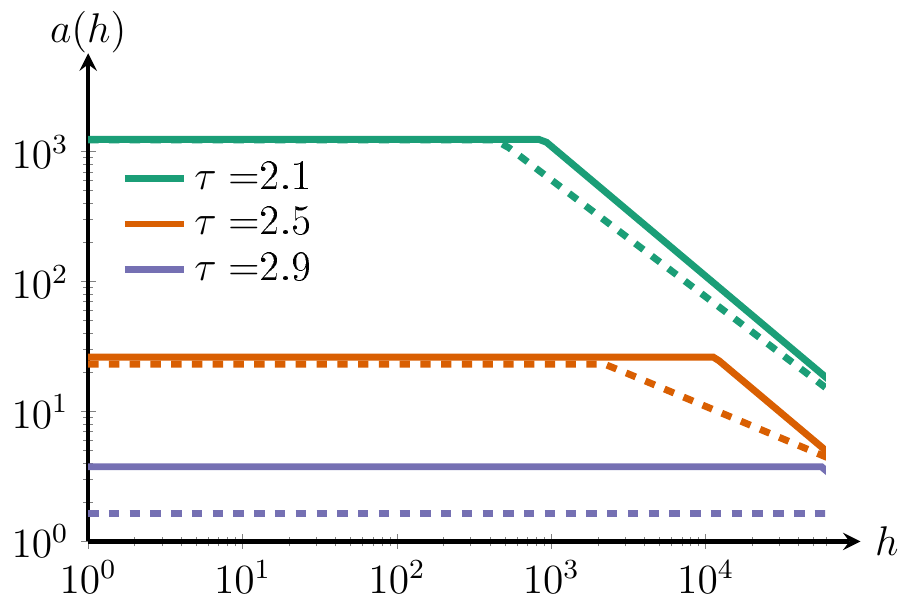}
    \label{fig:akplvar}
    }
    \hspace{1cm}
    \subfloat[]{
    \includegraphics[width=0.4\textwidth]{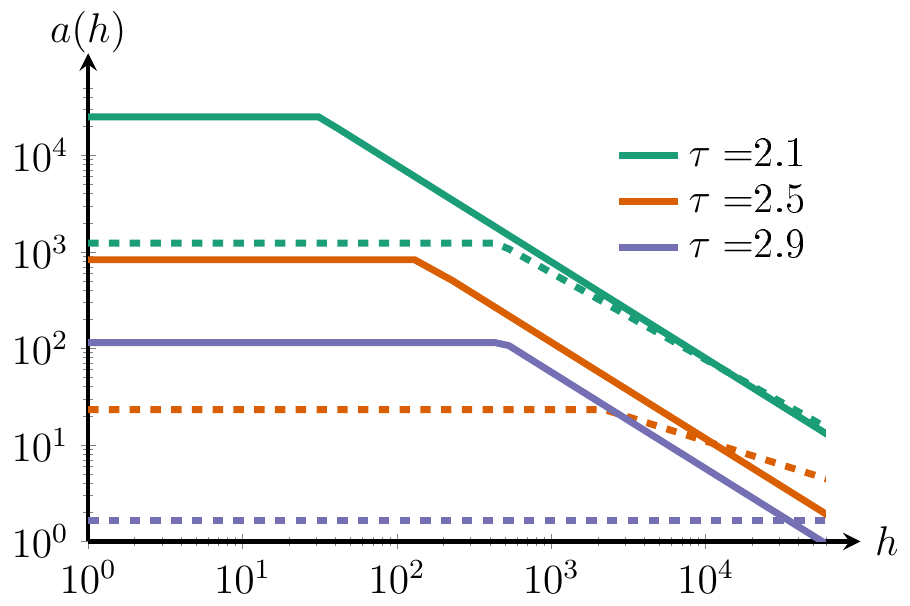}
    \label{fig:akplmad}
    }
    
        \subfloat[]
    {
    \includegraphics[width=0.4\textwidth]{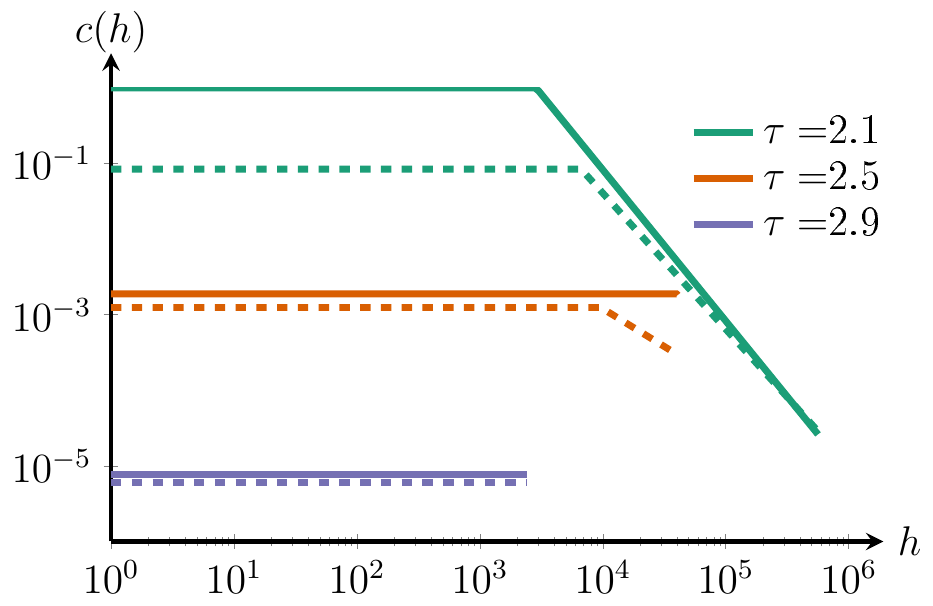}
    \label{fig:ckplvar}
    }
    \hspace{1cm}
    \subfloat[]{
    \includegraphics[width=0.4\textwidth]{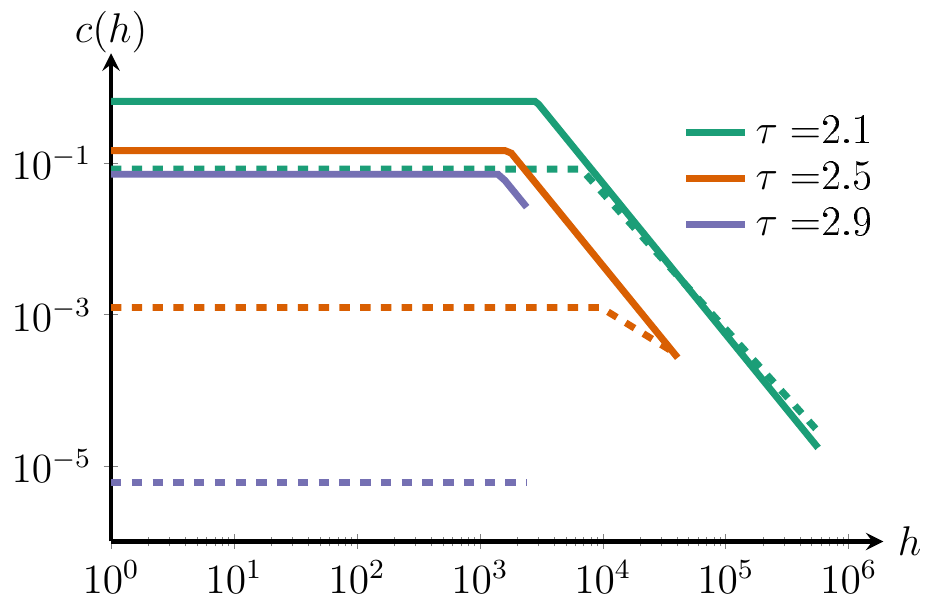}
    \label{fig:ckplmad}
    }
    \caption{Maximal scaling compared with power-law with the same parameters (dashed line) for $n=10^5$ on a) $a(k)$, variance basis (solid line), b) $a(k)$, MAD-basis (solid line) c) $c(k)$, variance basis (solid line), d) $c(k)$, MAD-basis (solid line).}
    \label{fig:akpl}
\end{figure*}

\paragraph{Comparing power-law $c(h)$ to extremal $c(h)$.}
We now turn to $c(k)$.
In power-law Chung-Lu models with cutoff $h_c=n^{1/(\tau-1)}$~\cite{stegehuis2017b}, 
\begin{equation}\label{eq:ckpl}
    c(h)\sim \begin{cases}
    n^{2-\tau}\ln(n) & h \ll n^{(\tau-2)/(\tau-1)}\\
    n^{2-\tau}\ln(n/h) & n^{(\tau-2)/(\tau-1)}\ll h \ll \sqrt{n}\\
    h^{2\tau-6}n^{5-2\tau} & h \gg\sqrt{n}.
    \end{cases}
\end{equation}
 We now compare this scaling to the scaling obtained by Theorem~\ref{thm:ckvariance}. Using~\eqref{eq:plsigma} shows that Theorem~\ref{thm:ckvariance} predicts that 
 \begin{equation}
     \max_{\pprob\in\mathcal{P}(\mu,n^{(3-\tau)/(\tau-1)})} \Expp{c(h)}\sim \min(n^{(7-3\tau)/(\tau-1)},1)
 \end{equation}
 for $h$ small, while it scales as $n/h^2$ for $h$ large. For $\tau=2$, this coincides with~\eqref{eq:ckpl}. Therefore, Theorem~\ref{thm:ckvariance} implies that the maximal $c(h)$ scaling among all Chung-Lu random graphs with given $\sigma^2$ and $\mu$ for $h\gg\sqrt{n}$ is achieved by the power-law distribution for $\tau=2$. However, for $h\ll\sqrt{n}$, the power-law distribution does not attain the largest possible value of $c(h)$. Figure~\ref{fig:ckplvar} illustrates this. The power-law $c(k)$ is even in its constant, very close to the variance-based maximal value of $c(k)$. For $\tau\approx 2$, the changing point between the constant regime and the decaying regime becomes close, while for larger values of $\tau$, the difference in changing point is larger. 



 For the MAD-based optimizer, we can obtain similar statements. Theorem~\ref{thm:ckhigh} predicts that 
  \begin{equation}
     \max_{\pprob\in\mathcal{P}^*(\mu,d)} \Expp{c(h)}\sim n/h^2
 \end{equation}
 for $h$ sufficiently high. Therefore, Theorem~\ref{thm:ckhigh} implies that the maximal $c(h)$ scaling among all Chung-Lu random graphs with given MAD, $h_c$ and $\mu$ for $h\gg\sqrt{n}$ is achieved by the power-law distribution for $\tau=2$. Indeed, Figure~\ref{fig:ckplmad} shows that the MAD-based optimal value of $c(k)$ is close to the power-law based one for $\tau\approx 2$, but that the bound can be far off otherwise. Furthermore, note that for $\tau=2.1$ the MAD-based optimal bound drops below the power-law achieved value. This is because of finite-size effects that are not included in Theorem~\ref{thm:ckhigh}.

\paragraph{Comparing power-law clique counts to extremal clique counts.}
We now compare the maximal amount of cliques predicted by Theorem~\ref{thm:Kkvar} to the amount of cliques achieved by a power-law random graph. When $2<\tau<3$, under a cutoff at $b=\sqrt{\mu n}$, the expected number of cliques in a power-law random graph with degree-exponent $\tau$ equals~\cite[Eq.~(1.7)]{janssen2019}
\begin{equation}\label{eq:Ekkcutoffpl}
	\expec_{pl}[N_{K_k}]\approx \frac{ n^{k/2(3-\tau)}\mu^{k/2(1-\tau)}}{k!}\left(\frac{C}{k-\tau}\right)^k.
\end{equation}
Now under a cutoff at $b=\sqrt{\mu n}$, for power-law random graphs, $\sigma^2=\frac{C}{3-\tau}\sqrt{\mu n}^{3-\tau}$. 
Plugging in $\sigma^2=\frac{C}{3-\tau}\sqrt{\mu n}^{3-\tau}$ from the power-law distribution into~\eqref{eq:Kkoptvarnocutoff} yields
\begin{equation}\label{eq:Ekkcutoff}
	\max_{\mathbb{P} \in \mathcal{P}^*_{(\mu,\sigma^2)}}\Expp{N_{K_k}}= \frac{ n^{k/2(3-\tau)}\mu^{k/2(1-\tau)}}{k!}\left(\frac{C}{3-\tau}\right)^k.
\end{equation}
This agrees in terms of scaling in $n$ and $\mu$ with the power-law Chung-Lu number of cliques scaling of~\eqref{eq:Ekkcutoffpl}.   So this proves that for variance-based clique optimization, power-laws contain the most number of cliques in terms of scaling in $n$ when using a cutoff at $b=\sqrt{\mu n}$. Still, the leading constant in~\eqref{eq:Ekkcutoff} is higher than the one in~\eqref{eq:Ekkcutoffpl} for $k>3$, so that the extremal random graph for cliques still achieves a higher total number of cliques in the leading order constant.

When $h_c=n^{1/(\tau-1)}$, then $\expec_{pl}[K_k]\sim n^{k(3-\tau)/2}$. Furthermore, in that setting, $\sigma^2\sim n^{(3-\tau)/(\tau-1)}$, while $\mu$ does not grow in $n$. Therefore, in this case, Theorem~\ref{thm:Kkvar} does not apply for $\tau<7/3$, as $\sigma^2\geq \sqrt{n} $ when $\tau<7/3$, so that the condition of Theorem~\ref{thm:Kkvar} does not apply. For $\tau>7/3$, plugging in the power-law value of  $\sigma^2\sim n^{(3-\tau)/(\tau-1)}$ into Theorem~\ref{thm:Kkvar} yields
\begin{equation}
    \max_{\pprob\in\mathcal{P}^*(\mu,\sigma^2)}\Expp{N(K_k)}\sim n^{k(3-\tau)/(\tau-1)},
\end{equation}
which is larger than the power-law scaling of $n^{k(3-\tau)/2}$. Therefore, power-law random graphs are not the graphs that contain the most cliques among all random graphs with the same mean and variance.

\paragraph{Data}\label{sec:data}
We now apply our bounds for $a(k)$ and $c(k)$ to three real-world network data sets. 
Figure~\ref{fig:aktot} compares the variance-based upper bound of Theorem~\ref{thm:anndvariance} with empirical observations. For all data sets, the true value of $a(k)$ exceeds the variance-based maximizer at some point. The MAD-based maximizer on the other hand, remains an upper bound for $a(k)$ in almost all data sets. This highlights the importance of the right choice of comparison model: The hidden-variable model with fitted variance cannot explain the degree-degree correlations in these data sets, while the same model with fitted MAD can. 
\begin{figure*}
\centering
    \subfloat[]
    {
    \centering
    \includegraphics[width=0.31\textwidth]{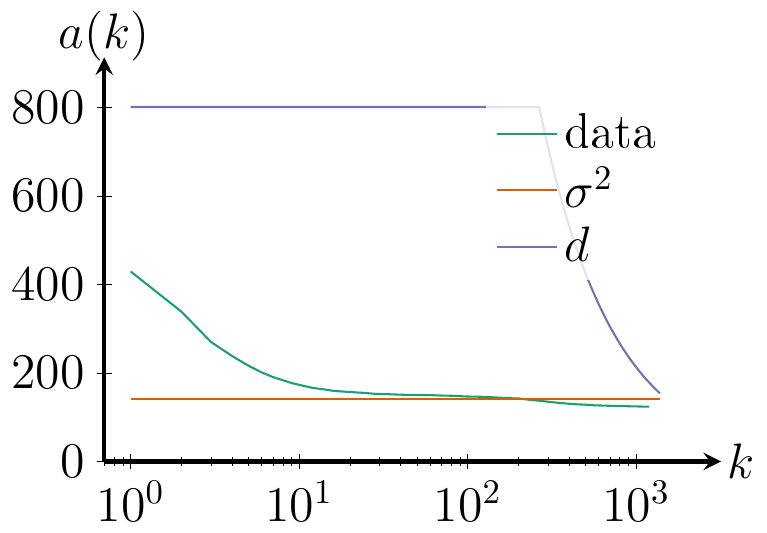}
    \label{fig:akenron}
    }
    \subfloat[]
    {
    \centering
    \includegraphics[width=0.31\textwidth]{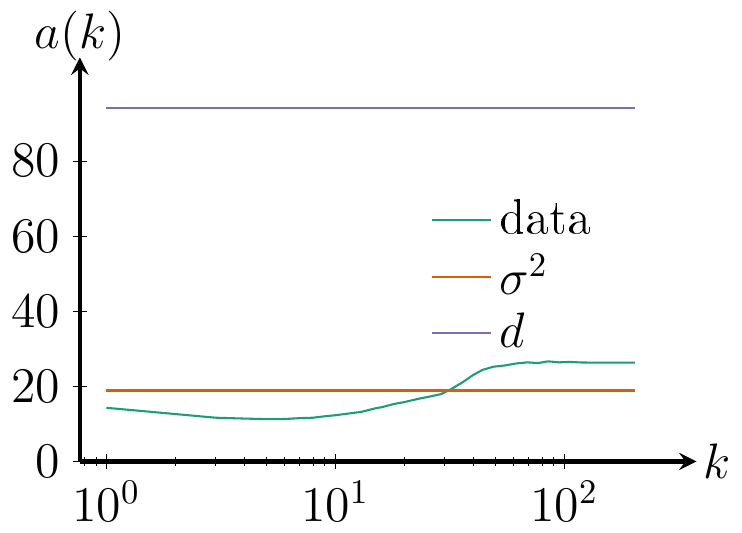}
    \label{fig:akpgp}
    }
     \subfloat[]
    {
    \centering
    \includegraphics[width=0.31\textwidth]{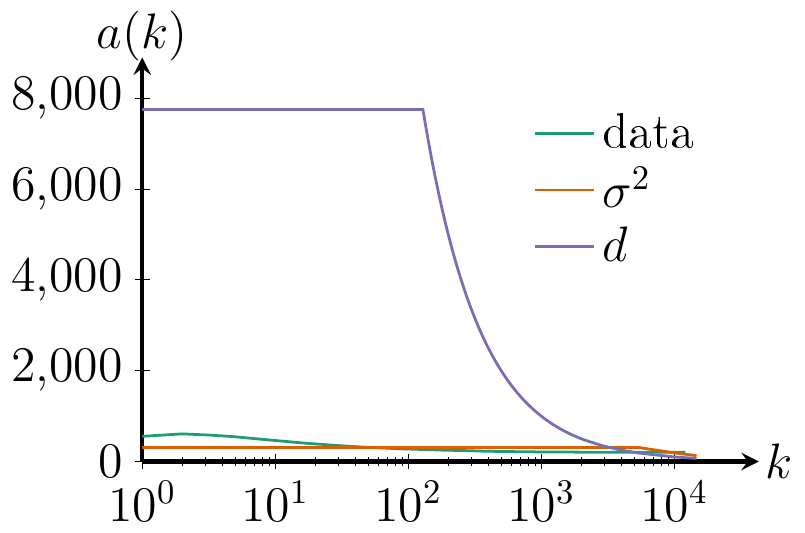}
    \label{fig:akgowalla}
    }
    \caption{$a(k)$ and MAD and variance-based bounds for 3 real-world networks a) Enron email network~\cite{klimt2004}, b) Pretty Good Privacy network~\cite{boguna2004a}, c) Gowalla social network~\cite{cho2011}.}
    \label{fig:aktot}
\end{figure*}

Figure~\ref{fig:cktot} shows for three real-world networks $c(k)$ and the variance- and MAD-based bounds. The $c(k)$-values of the Gowalla data set are close to, or below the MAD and the variance-based optimizer, respectively. This suggests that these data sets can be suitably modeled by some hidden-variable model that matches the $c(k)$ distribution of this data set. For the two other data sets on the other hand, the value of $c(k)$ in the data sets is higher than can be achieved by any hidden-variable model. Therefore, no hidden-variable model is able to match these data sets in terms of $c(k)$, which is likely caused by the locally-tree like nature of the hidden-variable model.

\begin{figure*}
\centering
    \subfloat[]
    {
    \centering
    \includegraphics[width=0.31\textwidth]{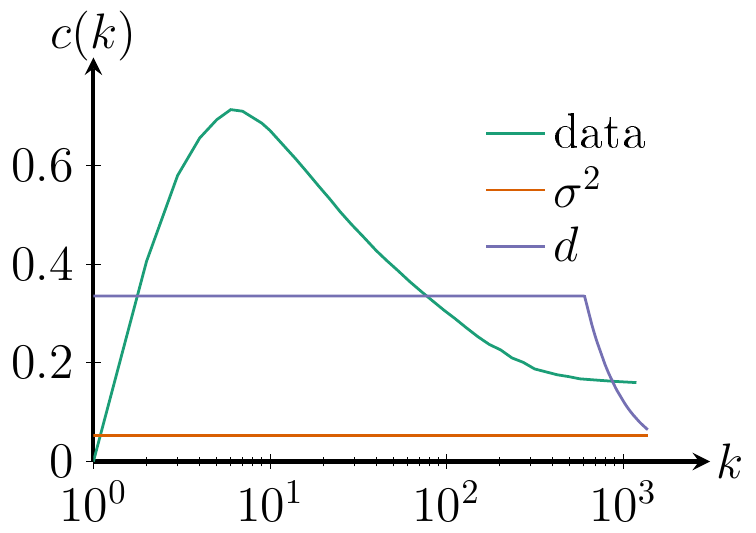}
    \label{fig:ckenron}
    }
    \subfloat[]
    {
    \centering
    \includegraphics[width=0.31\textwidth]{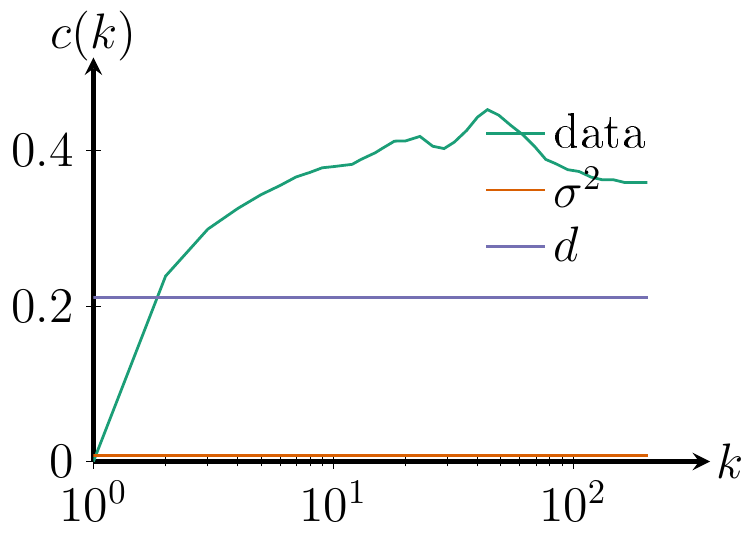}
    \label{fig:ckpgp}
    }
     \subfloat[]
    {
    \centering
    \includegraphics[width=0.31\textwidth]{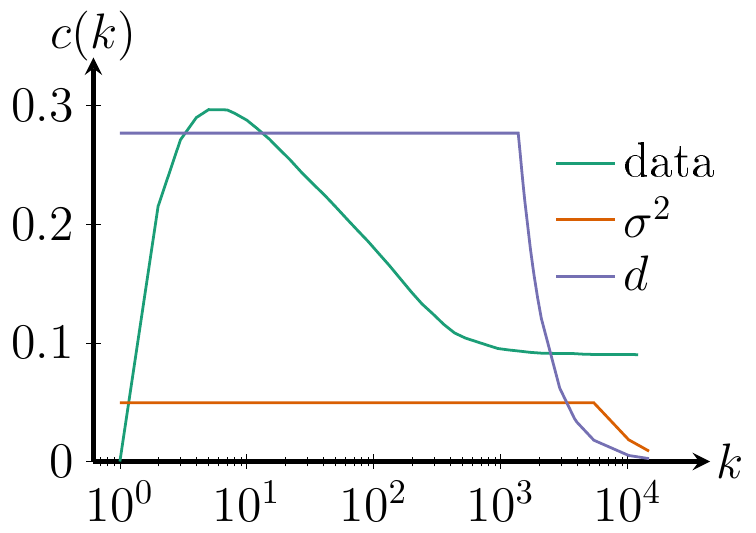}
    \label{fig:ckgowalla}
    }
    \caption{$c(k)$ and MAD and variance-based bounds for 3 real-world networks, a) Enron email network~\cite{klimt2004}, b) Pretty Good Privacy network~\cite{boguna2004a}, c) Gowalla social network~\cite{cho2011}.}
    \label{fig:cktot}
\end{figure*}

\section{Conclusions and discussion}

Our  robust network perspective, in terms of ambiguity and partial information on the degree distribution, comes with substantial mathematical challenges. We created an optimization framework for identifying, within some ambiguity set, the extreme degree distribution that generates the upper bound for the degree-degree correlations and clustering. We therefore had to combine probabilistic models (random graphs) with optimization models (stochastic programs). For successfully applying our robust perspective it is crucial to solve a given stochastic program in closed form. Here we distinguish between 1D programs such as the semi-infinite LP for the expected degree-degree correlation, and 2D programs such as for the expected clustering. For the 1D programs with both variance and MAD information, we were able to apply standard primal-dual techniques, solving the dual problem with the tightest majorant, and indeed finding a closed-form extremal distribution. 

The 2D programs for expected clustering proved to be more challenging, being semi-infinite programs with two  i.i.d.~random variables. This i.i.d.~assumption creates nonlinear conditions that prohibits the usage of standard primal-dual techniques. For variance information, we therefore applied a relaxation technique that replaces the original program with its counterpart that allows correlations. That relaxed program proved solvable with the primal-dual technique, despite the additional challenges of 2D instead of 1D. Surprisingly, the found extremal joint distribution was a product-form distribution after all, so that we could thus show that the relaxed program has the same solution as the original program. We showed that similar relaxations for programs of higher dimensions could be also be used to establish tight (asymptotic) bounds for clique counts.

However, such relaxations proved cumbersome, if not intractable, for MAD information. In that case, we opted for a different relaxation, in which we solve the 2D program in two steps: first finding the worst-case distributions of variable 1, and then, given this worst-case distribution, finding the worst-case distribution distribution of variable 2. This relaxation also allows correlation between the random variables, but possibly of a different nature.  This second type of relaxation turned out to give worst-case distributions that become 
independent (product-form) distributions when the network size grows to infinity. Hence, in this way, we could solve the original 2D MAD program  asymptotically for $n\to\infty$. 

This paper forms an important step towards a more complete theory of Distributionally Robust Random Graphs (DRRGs). This theory exchanges full information random graphs with partial information models, for instance regarding the degree distribution. This means that the theory applies to a large class of distributions, and can possibly explain complex network phenomena in a more universal manner, less dependent on the specific distributional assumptions. Hence, we consider this a robust way of studying complex networks. 

Here we mention a few open problems and research directions. For solving the 2D program we have introduced two relaxations, one for variance and one for MAD. Do both relaxations work for both variance and MAD? Can we understand when the two relaxations are equivalent, and when they are not? 
Do these relaxations also work in higher-dimensional stochastic programs, with more than two i.i.d.~random variables? In this paper we have successfully applied one such higher-dimensional relaxation for cliques of arbitrary size. 

Another avenue for future research concerns model extensions. 
We considered the Chung-Lu-type random graph, a tractable model that can accommodate a wide range of degree sequences. Drawbacks of this model are its locally tree-like nature and the slow convergence to the large-network limit~\cite{janssen2019}. 
%
This motivates to consider extended models such as random geometric graphs~\cite{penrose2003} or
generalizations of the popular hyperbolic random graph~\cite{krioukov2010,bringmann2015}, where the connection probability of two vertices scales as the product of their weights, divided by their distance. What is the maximal value of $c(k)$ for given mean and variance on the degrees and given mean and variance on the inter-distances? This question will lead to a more involved optimization problem with more variables due to the underlying geometry. 
Such upper bounds for increasing model complexity could detect what level of complexity is necessary to model a specific network property correctly. 

Furthermore, while we investigated robust degree distributions, one can also think of other network properties. For example, in temporal network models one can obtain results for robust edge time-stamps. For hypergraphs, one can think of robust hyperdegrees, or of robust positions for geometric models. We believe that this framework can provide robust upper bounds on several network properties, and quantify the sensitivity of network models to specific assumptions on their parameters. 

Finally, it is
worthwhile to compare extremal random graphs as studied in this paper with entropy-maximizing random graph ensembles~\cite{park2004}.
While we maximize a given network property (such as clustering or degree-correlations), maximum entropy random graphs maximize the entropy instead, and then compute the value of the property. This calls to investigate for classes of degree distributions and network properties whether extremal random graph and maximum entropy random graphs are comparable or not.

\acknowledgments{We thank Dick den Hertog, Wouter van Eekelen and Pieter Kleer for various thought exchanges about the relaxations introduced in this paper for solving semi-infinite linear programs, and robust optimization in general. JB is supported by an NWO Mathematics Clusters grant, JvL is supported by an NWO VICI grant. CS is supported by NWO VENI grant 202.001 and NWO M2 grant 0.379.}

\bibliographystyle{abbrv}
\bibliography{references}

\appendix
\section{Proofs}

\subsection{Proof of the variance-based maximizer for $c(k)$ and the number of cliques}\label{sec:ckvarbasedproof}
\begin{proof}[Proof of Theorem~\ref{thm:ckvariance}]
The optimization problem~\eqref{eq:varoptdualch} is equivalent to 
\begin{equation}\label{eq:varoptdualchtemp}
\begin{aligned}
&\max_{\pprob(x)\geq0} &  &\int_{x_1}\int_{x_2} g(x_1,x_2){\rm d} \pprob(x_2){\rm d} \pprob(x_1)\\
&\text{s.t.} &      & \int_{x_1}\int_{x_2} x_1^2x_2^2{\rm d}\pprob(x_2){\rm d}\pprob(x_1)=(\mu^2+\sigma^2)^2\\
& & & \int_{x_1}\int_{x_2} x_1x_2{\rm d}\pprob(x_2){\rm d}\pprob(x_1)=\mu^2,\\
&  & & \int_{x_1}\int_{x_2} {\rm d}\pprob(x_2){\rm d}\pprob(x_1)=1.   
\end{aligned}
\end{equation}
We now consider a relaxed version of this optimization problem. Instead of drawing from a single measure $\pprob$ for both $x_1$ and $x_2$, we allow for a dual measure $\pprob(x_1,x_2)$, where we only require the product of the means and second moments to be equal to $\mu^2$ and $(\mu^2+\sigma^2)^2$, respectively. We thus drop the i.i.d. assumption for now. This gives the problem
\begin{equation}\label{eq:varoptrelaxed}
\begin{aligned}
&\max_{\pprob(x_1,x_2)\geq0} &  &\int_{x_1}\int_{x_2} g(x_1,x_2){\rm d} \pprob(x_1,x_2)\\
&\text{s.t.} &      & \int_{x_1}\int_{x_2} x_1^2x_2^2{\rm d}\pprob(x_1,x_2)=(\mu^2+\sigma^2)^2 \\
&  &  & \int_{x_1}\int_{x_2} x_1x_2{\rm d}\pprob(x_1,x_2)=\mu^2,\\
&  & & \int_{x_1}\int_{x_2} {\rm d}\pprob(x_1,x_2)=1.   
\end{aligned}
\end{equation}
The dual problem then becomes
\begin{equation}\label{eq:varoptch}
\begin{aligned}
&\min_{\lambda_1,\lambda_2,\lambda_3} &  &\lambda_1(\mu^2+\sigma^2)^2+\lambda_2\mu^2+\lambda_3\\
&\text{s.t.} &      & g(x_1,x_2)-\lambda_1x_1^2x_2^2-\lambda_2x_1x_2-\lambda_3\leq 0\nonumber\\
&& &  \forall x_1,x_2\in[a,h_c],
\end{aligned}
\end{equation}
with 
\begin{equation}
g(x_1,x_2)=\min\Big(\frac{x_1x_2}{\mu n},1\Big)\min\Big(\frac{x_1h}{\mu n},1\Big)\min\Big(\frac{x_2h}{\mu n},1\Big).\end{equation} 

We will now solve the optimization problem by constructing a primal and dual solution that achieve the same objective value, and therefore optimize~\eqref{eq:varoptrelaxed}. Furthermore, the constructed optimal probability distribution turns out to be of product form, so that they must also be optimizers of the original, more constrained optimization problem~\eqref{eq:varoptdualchtemp}. These primal and dual solutions depend on $h, n, \mu$ and $\sigma$, in the following cases:

\textit{Case 1: $h\leq\sqrt{\mu n}$.} We take the dual solution $\lambda_1=h^2/(\mu n)^3$, $\lambda_2=\lambda_3=0$. This gives as objective value $(\mu^2+\sigma^2)^2h^2/(\mu n)^3$.

When $\sigma^2\leq (\sqrt{\mu n}-\mu)\mu$, for the primal problem, consider the 3-point distribution
\begin{align}
    p_0& =\frac{\sigma^2}{\sqrt{\mu n}\mu},\quad p_\mu=1-\frac{\sigma^2}{\mu(\sqrt{\mu n} - \mu)},\nonumber\\
    p_{\sqrt{\mu n}}&=\frac{\sigma^2}{\sqrt{\mu n}(\sqrt{\mu n}-\mu)}.
\end{align}
This is  a proper distribution by the condition on $\sigma^2$.
Then,
\begin{align}
    &\Exp{g(X_1,X_2)}\nonumber\\
    &=\frac{h^2}{(\mu n)^3}\Big(p_\mu^2\mu^4+2p_\mu p_{\sqrt{\mu n}}\mu^2(\mu n)+p_{\sqrt{\mu n}}^2(\mu n)^2\Big)\nonumber\\
    & = \frac{h^2}{(\mu n)^3}\Big(p_\mu\mu^2+p_{\sqrt{\mu n}}\mu n\Big)^2 = \frac{h^2}{(\mu n)^3}(\mu^2+\sigma^2)^2.
\end{align}
Thus, by strong duality, this is the optimizer of~\eqref{eq:varoptdualch}. 

When $\sigma^2 > (\sqrt{\mu n}-\mu)\mu$, consider the three-point distribution
\begin{align}\label{eq:primal-dist}
    p_0& =1-p_{\sqrt{\mu n}}-p_{\mu n/h},\nonumber\\
    p_{\sqrt{\mu n}}&= \frac{\mu ^2+\sigma ^2- \mu n/h \cdot \mu}{\sqrt{\mu n} (\sqrt{\mu n} - \mu n / h)},\nonumber\\
    p_{\mu n/h}& =\frac{\mu ^2+\sigma ^2- \sqrt{\mu n} \cdot \mu}{\mu n /h ( \mu n / h - \sqrt{\mu n})}.
\end{align}
This is only a proper distribution when $\sigma^2<(\mu n /h - \mu)\mu$.
This three-point distribution gives $\Exp{c(h)}=1$, so that it achieves the maximum $c(h)$. Therefore, we can immediately conclude that this primal solution is optimal.

\textit{Case 2: $\sqrt{\mu n}<h<\mu^2n/(\mu^2+\sigma^2)$.}
We take as dual solution again $\lambda_1=h^2/(\mu n)^3$, $\lambda_2=\lambda_3=0$. This gives as objective value $(\mu^2+\sigma^2)^2h^2/(\mu n)^3$.

For the primal problem, consider the 3-point distribution
\begin{align}
    p_0& =\frac{\sigma^2}{\mu(\mu n/h)}, \quad p_\mu =1-\frac{\sigma^2}{\mu(\mu n/h - \mu)},\nonumber\\
    p_{\mu n/h}&=\frac{\sigma^2}{\mu n/h(\mu n/h-\mu)}.
\end{align}
Again, this is only a distribution when $\sigma^2\leq\mu(\mu n/h-\mu)$, which is ensured by the condition on $h$. 
Then,
\begin{align}
    &\Exp{g(X_1,X_2)}\nonumber\\
    &=\frac{h^2}{(\mu n)^3}\Big(p_\mu^2\mu^4+2p_\mu p_{\mu n/h}\mu^2(\mu n)^2/h^2+p_{\mu n/h}^2(\mu n/h)^4\Big)\nonumber\\
    & = \frac{h^2}{(\mu n)^3}\Big(p_\mu\mu^2+p_{\mu n/h}(\mu n/h)^2\Big)^2 = \frac{h^2}{(\mu n)^3}(\mu^2+\sigma^2)^2.
\end{align}
Thus, the primal solution achieves the same value as the dual solution. Therefore, by strong duality, this is the optimizer of~\eqref{eq:varoptdualch}.

\textit{Case 3:  $h\geq \mu^2n/(\mu^2+\sigma^2)$.} We take as dual solution $\lambda_2=1/(\mu n)$, $\lambda_1=\lambda_3=0$. This gives as objective value $\mu/n$.

For the primal problem, consider again the 3-point distribution that is given in (\ref{eq:primal-dist}).
This is only a proper distribution when $\sigma^2>(\mu n /h - \mu)\mu$, which is satisfied by our condition on $\sigma^2$, and $\sigma^2\leq (\sqrt{\mu n}-\mu)\mu$. Under this three-point distribution, $\mathbb{E}[X]=\mu$, $\mathbb{E}[(X-\mu)^2]=\sigma^2$ and $\mathbb{E}[g(X_1,X_2)]=\mu/n$. Thus, by strong duality, this is the optimal solution.
\end{proof}

\begin{proof}[Proof of Theorem~\ref{thm:Kkvar}]
The relaxed optimization problem corresponding to~\eqref{eq:varoptdualclique} gives the dual problem
\begin{equation}\label{eq:varoptclique}
\begin{aligned}
&\min_{\lambda_1,\lambda_2,\lambda_3}  & & \lambda_1(\mu^2+\sigma^2)^k+\lambda_2\mu^k+\lambda_3\\
&\text{s.t.~}    & &  g(x_1,\dots,x_k)-\lambda_1x_1^2\cdots x_k^2-\lambda_2x_1x_2\cdots x_k-\lambda_3\leq 0 \nonumber\\
& & &\quad \forall x_1,\dots,x_k\in[a,h_c],
\end{aligned}
\end{equation}
with 
\begin{equation}
g(x_1,\dots,x_k)=\prod_{1\leq i<j\leq k}\min\Big(\frac{x_ix_j}{\mu n},1\Big).\end{equation} 
Consider the dual solution $\lambda_1=1/(\mu n)^k$, giving as objective value $(\mu^2+\sigma^2)^k/(\mu n)^k$. 

Consider the 3-point distribution 
\begin{align}
    p_0&=1-p_m-p_{\sqrt{\mu n}} ,\quad p_m=\frac{\mu \sqrt{\mu n}-\mu^2-\sigma^2}{m(\sqrt{\mu n}-m)},\nonumber\\ p_{\sqrt{\mu n}}&=\frac{\mu^2+\sigma^2- m \mu }{\sqrt{\mu n}(\sqrt{\mu n}-m)},
\end{align}
with $m=\mu^{(1+k)/4}n^{(3-k)/4}$. 
Note that $m<\sqrt{\mu n}$ for $k>3$, and $m=\mu$ for $k=3$, and that this is only a proper distribution when $\sigma^2\leq \mu(\sqrt{\mu n}-\mu)$. Now 
\begin{align}
    &\Exp{g(X_1,\dots,X_k)}=\frac{\Exp{X_1^{k-1}}^k}{(\mu n)^{k(k-1)/2}}. 
\end{align}
Furthermore,
\begin{align}
    \Exp{X_1^{k-1}}&=\frac{\left(\mu ^{\frac{k+1}{4}} n^{\frac{3-k}{4}}\right)^{k-2} \left(\mu \sqrt{\mu  n}-\mu^2-\sigma ^2\right)}{\sqrt{\mu  n}-\mu ^{\frac{k+1}{4}} n^{\frac{3-k}{4}}}\nonumber\\ 
    & + \frac{\sqrt{\mu  n}^{k-2} \left(\mu ^2+\sigma ^2-\mu ^{\frac{k+5}{4}} n^{\frac{3}{4}-\frac{k}{4}}\right)}{\sqrt{\mu  n}-\mu ^{\frac{k+1}{4}} n^{\frac{3-k}{4}}}.
\end{align}
Now when $k> 3$, then $n^{(3-k)/4}=o(1)$. Therefore, for $k>3$, 
\begin{align}
    \Exp{X_1^{k-1}}&= \frac{\sqrt{\mu  n}^{k-2} \left(\mu ^2+\sigma ^2\right)}{\sqrt{\mu  n}}(1+o(1)).
\end{align}
Thus, also
\begin{align}
    &\Exp{g(X_1,\dots,X_k)}
    = \frac{\sqrt{\mu n}^{(k-3)k}(\mu^2+\sigma^2)^k}{(\mu n)^{k(k-1)/2}}(1+o(1))\nonumber\\
    & = \frac{(\mu^2+\sigma^2)^k}{(\mu n)^{k}}(1+o(1)),
\end{align}
making the 3-point distribution asymptotically optimal, as it asymptotically achieves the same value as the dual solution.
Furthermore, for $k=3,$ $\Exp{X^2}=\mu^2+\sigma^2$, by the conditions in $\mathcal{P}(\mu,\sigma^2)$. Thus, for $k=3$,
\begin{align}
    &\Exp{g(X_1,\dots,X_k)}=\frac{(\mu^2+\sigma^2)^3}{(\mu n)^{3}},
\end{align}
which is the exact same value as the dual objective value.
Hence, by strong duality, this is the optimal solution.
\end{proof}

\subsection{Proofs for MAD-based maximizers of $a(k)$ and $c(k)$}\label{sec:proofsMAD}
\begin{proof}[Proof of Theorem~\ref{thm:annd}]

The function $h'\min(\frac{h h'}{\mu n},1)$ is piecewise convex in $h'$. In particular, it is quadratic up to $l=\min(\mu n/h, h_c)\gg \mu$, where it has slope 2, and it is linear with slope 1 after that. 

Thus, to optimize over the distribution of $h'$, we need to solve
\begin{align}\label{eq:primal}
   &\max_{\pprob\in\mathcal{P}(\mu,d)}\int_x f(x) {\rm d}\mathbb{P}(x)\nonumber\\
   &\text{s.t. }  \ \int_x |x-\mu|{\rm d}\mathbb{P}(x)=d, \int_x x{\rm d}\mathbb{P}(x)=\mu, \int_x {\rm d}\mathbb{P}(x)=1,
\end{align}
where $f(x)=x\min(\frac{h x}{\mu n},1)$.
Similarly to the derivation of~\cite[Eq. (6)]{eekelen2019}, this results in the dual problem
\begin{align}\label{eq:dual}
    \min_{\lambda_1,\lambda_2,\lambda_3}&\lambda_1d+\lambda_2\mu +\lambda_3\nonumber\\
\text{s.t. }& f(x)-\lambda_1|x-\mu|-\lambda_2 x-\lambda_3\leq 0 \quad \forall x\in[0,h_s].
\end{align}
For simplicity of notation, we assume that $a=0$. 
Thus, this dual problem aims to find the tightest piecewise linear majorant of $f(x)$ with a kink at $\mu$ that minimizes the objective value. 
Consider the majorant $F_1(x)=\frac{h l}{2\mu n}|x-\mu|+(\frac{hl}{2\mu n}+\frac{h}{n})x-\frac{h l}{2n}$.
Now $F_1(x)$ has as its objective value 
\begin{align}&\frac{h l}{2\mu n}d+(\frac{hl}{2\mu n}+\frac{h}{n})\mu-\frac{h l}{2n} 
=\frac{\mu h}{n}+\frac{hdl}{2\mu n}.\end{align}

By weak duality of semi-infinite linear programming, we know that a feasible solution to the dual problem provides us with a valid upper bound for the optimal primal solution value. Thus, we now find a feasible primal solution with an objective value equal to this upper bound results to achieve strong duality.
As the tightest majorant of the dual problem touches $f(x)$ at the points $a,\mu$ and $l=\min(\mu n/h, h_c)$, we consider the three-point distribution 
\begin{align}\label{eq:3pointplow}
	p_0& =\frac{d}{2\mu},\quad p_\mu=1-\frac{d}{2\mu}-\frac{d}{2(l-\mu)},\nonumber\\
	\quad p_{l}&=\frac{d}{2(l-\mu)},
\end{align}
which is a distribution since $\mu n/h\gg \mu$.
This yields as objective value for the primal problem
\begin{align}\label{eq:primalboundak}
    &\frac{\mu^2h}{\mu n}\Big(1-\frac{d}{2\mu}-\frac{d}{2(l-\mu)}\Big) +\frac{l^2h}{\mu n}\frac{d}{2(l-\mu)}=\frac{\mu h}{n}+\frac{hdl}{2\mu n}.
\end{align}

Thus, we have strong duality as the primal objective from~\eqref{eq:primalboundak} and the dual optimal value are the same. 

Therefore,
\begin{align}
    &\Expp{a(h)}  = \frac{n}{h} \Big(\frac{\mu h}{n}+\frac{hdl}{2\mu n}\Big) = \mu + \frac{dl}{2\mu}.
\end{align}
Now when $h_c$ tends to infinity the second term dominates as $l=\min(h_c,\mu n/h)\gg1$ when $h\ll n$ and,
\begin{align}
    \Expp{a(h)}     & = \frac{d}{2\mu}\min(h_c,\frac{\mu n}{h})(1+o(1)).
\end{align}
\end{proof}

\begin{proof}[Proof of Theorem~\ref{thm:ckhigh}]
\textit{Case 1: $h\ll\sqrt{\mu n}$ and $2\mu(1-\sqrt{\mu/n})<d<2\mu(1-h/n)$.} In this case, consider the three-point distribution
\begin{align}
    p_0&=\frac{d}{2\mu},\quad 
    p_{\sqrt{\mu n}}= \frac{(d-2 \mu) l+2 \mu ^2}{2 \mu  \left(\sqrt{\mu  n} - l\right)},\nonumber\\
    p_{l}&=\frac{(d -2\mu)\sqrt{\mu  n}+2 \mu ^2}{2 \mu  \left(l-\sqrt{\mu  n}\right)},
\end{align}
for $l= \mu n/h$,
which is a proper distribution under the condition $2\mu(1-\sqrt{\mu/n})<d<2\mu(1-h/n)$. 
For this three-point distribution,
\begin{equation}
    \Expp{c(h)}=\frac{n^2}{h^2}\Big(p_{\sqrt{\mu n}}^2\frac{\mu n h^2 }{(\mu n)^2}+2p_lp_{\sqrt{\mu n}}\frac{\sqrt{\mu n} h}{\mu n}+p_{l}^2\Big)=1.
\end{equation}
As $c(h)$ is a probability, we have $\Exp{c(h)}\leq 1$, so that it coincides with the upper bound.

\textit{Case 2: $h\gg\sqrt{\mu n}$.} We now apply the three-point optimization problem in two steps: first for the optimal distribution of $h_i$, then for $h_j$. We will show that these two optimal distributions are identical and independent, so that this method shows that optimizing the distribution of $h_i$ and $h_j$ while constraining them to be equal yields this same optimal distribution. 

The function we would like to optimize is
\begin{equation}\label{eq:opteqck}
    \Exp{\min(\frac{hh_1}{\mu n},1)\min(\frac{hh_2}{\mu n},1)\min(\frac{h_1h_2}{\mu n},1)}.
\end{equation}

\textit{Step 1: optimizing over $h_1$.}
If we optimize only over the distribution of $h_1$ and fix $h_2$, this is equivalent to optimizing 
\begin{equation}\label{eq:maxch1}
    \Exp{\min(\frac{hh_1}{\mu n},1)\min(\frac{h_1h_2}{\mu n},1)}.
\end{equation}
Thus, we again want to maximize~\eqref{eq:primal} and therefore minimize its dual problem~\eqref{eq:dual}, but now with $f(x)=\min(\frac{hx}{\mu n},1)\min(\frac{xh_2}{\mu n},1)$. We again focus on the dual problem, which is again equivalent to minimizing $\lambda_1 d+\lambda_2\mu +\lambda_3$ over all tightest majorants of $f(x)$. 
Now $f(x)$ is again piecewise convex: it is quadratic in $h_1$ for $h_1<\min(\mu n/h,\mu n/h_2)$, linear for $h_1\in[\min(\mu n/h,\mu n/h_2),\max(\mu n/h,\mu n/h_2)]$, and constant for $h_1\in[\max(\mu n/h,\mu n/h_2),h_c]$. Thus, $f(x)$ is shaped as the function in Figure~\ref{fig:majorantck1} (where for simplicity $a=0$). 

\begin{figure}[tb]
    \centering
    \includegraphics[width=0.45\textwidth]{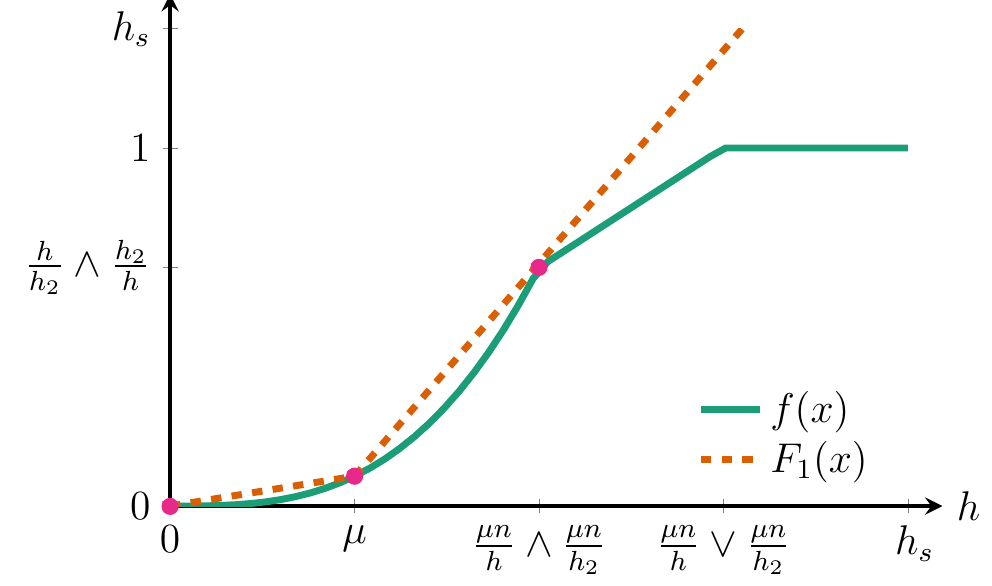}
    \caption{The tightest majorant for $f(x)$ describing the optimizer over $h_1$ of $c(h)$ }
    \label{fig:majorantck1}
\end{figure}

In the next computations, we assume for simplicity of notation that $a=0$. 
Now the tightest majorant of Figure~\ref{fig:majorantck1}, $F_1(x)$, can be parametrized by $\lambda_1=\min(h,h_2)/(2\mu n)$, $\lambda_2=\min(h,h_2)/(2\mu n)+hh_2/(\mu n^2)$, $\lambda_3=-\min(h,h_2)/(2n)$. This gives as objective value for the dual program~\eqref{eq:dual}
\begin{align}\label{eq:chdualub}
& d\frac{\min(h,h_2)}{2\mu n}+\mu\Big(\frac{\min(h,h_2)}{2\mu n}+\frac{hh_2}{\mu n^2}\Big)-\frac{\min(h,h_2)}{2n}\nonumber\\
& =  \frac{d}{2\mu n}\min(h,h_2)+\frac{hh_2}{n^2}.
\end{align}

We now again consider the primal problem~\eqref{eq:primal}. The solution to the dual problem~\eqref{eq:dual}, $F_1(x)$ has three touching points of $f(x)$: at $0$, $\mu$ and $\min(\mu n/h_2,\mu n/h)$. Thus, we will now show that $c(h)$ is maximized over $h_1$ by the three-point distribution 
\begin{align}\label{eq:3pointc1}
	&p_0 =\frac{d}{2\mu}, \quad p_\mu=1-\frac{d}{2\mu}-\frac{d}{2(\min(\mu n/h_2,\mu n/h)-\mu)},\nonumber\\
	&p_{\min(\mu n/h_2,\mu n/h)}=\frac{d}{2(\min(\mu n/h_2,\mu n/h)-\mu)}.
\end{align}
This gives an objective value of~\eqref{eq:primal} of 
\begin{align}
 &0\cdot \frac{d}{2\mu} + \frac{\mu^2hh_2}{(\mu n)^2}\Big(1-\frac{d}{2\mu}-\frac{d}{2(\min(\mu n/h_2,\mu n/h)-\mu)}\Big)\nonumber\\
 & +\frac{\min(\mu n/h_2,\mu n/h)^2hh_2}{(\mu n)^2}\cdot\frac{d}{2(\min(\mu n/h_2,\mu n/h)-\mu)}\nonumber\\
 & =\frac{d}{2\mu n}\min(h,h_2)+\frac{hh_2}{n^2}.
\end{align}
Thus, the objective value of the three-point distribution is equal to the objective value of the dual problem in~\eqref{eq:chdualub}. Thus, by strong duality,~\eqref{eq:3pointc1} is the optimal distribution maximizing~\eqref{eq:maxch1}.

\textit{Step 2: optimizing over $h_2$.}
We now plug the optimal three-point distribution of $h_1$~\eqref{eq:3pointc1}  into~\eqref{eq:opteqck} and then optimize only over the distribution of $h_2$. We then need to optimize
\begin{widetext}
\begin{align}\label{eq:exph2}
    &\Exp{(1-\frac{d}{2\mu}-\frac{d}{2(\min(\mu n/h_2,\mu n/h)-\mu)})\frac{h h_2}{ n^2}\min(\frac{hh_2}{\mu n},1)} + \Exp{\frac{d}{2 (\min( \mu n/h_2, \mu n/h)-\mu)}\min(\frac{h}{h_2},\frac{h_2}{h})\min(\frac{hh_2}{\mu n},1)}.
\end{align}
\end{widetext}

This is a function in $h_2$ that looks like Figure~\ref{fig:majorantck2}: first a convex part, then a linear part, and then again a convex part. 
The tightest majorant $F_1(x)$ is depicted in Figure~\ref{fig:majorantck2} as well.
\begin{figure}[tb]
    \centering
    \includegraphics[width=0.4\textwidth]{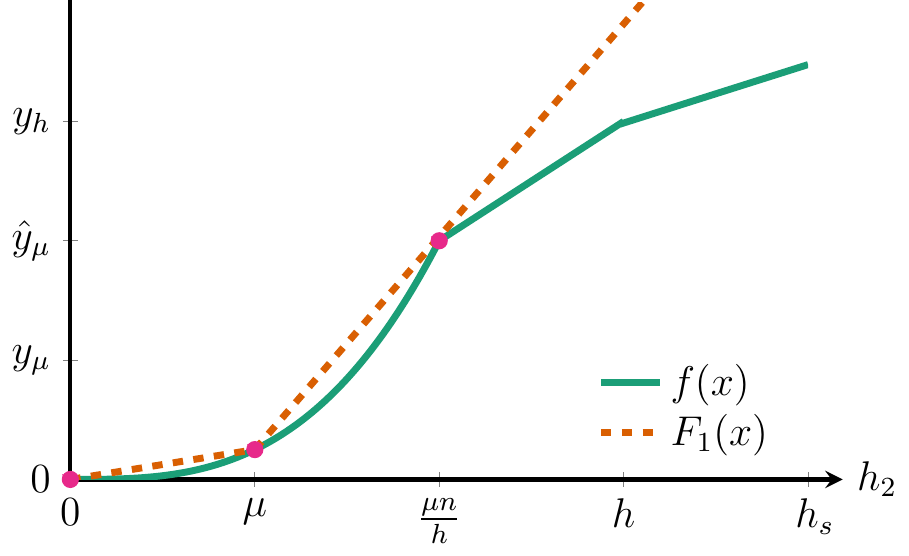}
    \caption{The tightest majorant for $f(x)$ describing the optimizer over $h_2$ of $c(h)$}
    \label{fig:majorantck2}
\end{figure}
$F_1(x)$ is characterized by $\lambda_1=\tilde{C}h/(2(n\mu)^2) $, $\lambda_2=\tilde{C}(h^2/(\mu^2n^3)-h/(2(n\mu)^2))$ and $\lambda_3=-\tilde{C}h/(2n^2\mu)$, with $\tilde{C}=\mu^2+dn\mu/(2h)$. This gives as dual objective  
\begin{equation}
    \lambda_1d+\lambda_2\mu+\lambda_3=\frac{(d n+2 h \mu )^2}{4 \mu  n^3}.
\end{equation}
We now consider the 3-point distribution for $h_2$ on the touching points $0, 
\mu$ and $\mu n/h$:  
\begin{align}\label{eq:3pointc2}
	p_0& =\frac{d}{2\mu},\quad p_\mu=1-\frac{d}{2\mu}-\frac{d}{2 (\mu n/h-\mu)},\nonumber\\
	\quad p_{\mu n/h}& =\frac{d}{2(\mu n/h-\mu)}.
\end{align}
This gives as objective value for the primal problem
\begin{align}
    &\Exp{\Big(1-\frac{d}{2\mu}-\frac{d}{2(\mu n/h-\mu)}\Big)\frac{h^2}{\mu n^3}h_2^2} + \Exp{\frac{d}{2 ( \mu n/h-\mu)}\frac{h_2^2}{\mu n}}\nonumber\\
    &=\frac{(d n+2 h \mu )^2}{4 \mu  n^3},
\end{align}
so that by strong duality, this is the optimal three-point distribution for $h_2$.
As $h\gg\sqrt{\mu n}$, this means that $h_2<h$ for all three values of the three-point distribution. Then, the three-point distribution for~\eqref{eq:3pointc1} reduces to~\eqref{eq:3pointc2}. 

Thus, by optimizing the distributions of $h_1$ and $h_2$ separately, we obtain the same three-point distribution for both. Therefore, the optimization of the distributions of $h_1$ and $h_2$ where they are constrained to have the same distribution also gives the three-point distribution~\eqref{eq:3pointc2} as its solution.

Indeed, 
\begin{equation}
    \max_{x,y:x=y}f(x,y)\leq \max_y \max_x f(x,y),
\end{equation}
as all combinations $f(x,x)$ are also encountered on the right-hand side.
Furthermore, let $x^*$ and $y^*$ denote the optimizers obtained in the right-hand side, and suppose that $x^*=y^*$. Then,
\begin{equation}
    \max_{x,y:x=y}f(x,y)\geq f(x^*,y^*)= \max_y \max_x f(x,y).
\end{equation}
Thus, when optimizing the distributions of $h_1$ and $h_2$ separately yields an optimizer where both distributions are equal, then this is also the optimization of the distributions of $h_1$ and $h_2$ where they are constrained to have the same distribution. 

This gives for $c(h)$
\begin{align}
     \max_{\pprob\in\mathcal{P}(\mu,d)}\Expp{c(h)} &= \frac{n^2}{h^2}\frac{(d n+2 h \mu )^2}{4 \mu  n^3} \nonumber\\
     &=\frac{d^2n}{4\mu h^2}(1+o(1)).
\end{align}

\textit{Case 3: $h\ll\sqrt{\mu n}$ and $d<2\mu(1-\mu/h)$.} We optimize $c(h)$ here by again first optimizing over the distribution of $h_i$ only, and then over the distribution of $h_j$. Therefore, up until~\eqref{eq:exph2} we follow the same steps for optimizing over $h_i$. We then optimize for the distribution of $h_j$. Consider the majorant of the dual problem $\tilde{F}_1(x)$ described by $\lambda_1=\frac{h^2 (d n+2 h \mu )}{4 \mu ^2 n^3}$, $\lambda_2=\frac{h (h+2 \mu ) (d n+2 h \mu )}{4 \mu ^2 n^3}$ and $\lambda_3=-\frac{h^2 (d n+2 h \mu )}{4 \mu  n^3}$, giving as dual objective value
\begin{equation}\label{eq:dualoptcase3}
    \lambda_1d+\lambda_2\mu+\lambda_3=\frac{h \left(d h+2 \mu ^2\right) (d n+2 h \mu )}{4 \mu ^2 n^3}.
\end{equation}
For the primal problem, consider the three-point distribution for $h_2$ of
\begin{align}
    p_0&=\frac{d}{2\mu},\quad p_\mu=1-\frac{d}{2\mu}-\frac{d}{2(h-\mu)},\nonumber\\
    \quad p_{h}&=\frac{d}{2 (h-\mu )},
\end{align}
which is a proper distribution as long as $d<2\mu(1-\mu/h)$.
Plugging this into~\eqref{eq:exph2} for the distribution of $h_j$ gives
\begin{align}
    &\Exp{\Big(1-\frac{d}{2\mu}-\frac{d}{2(\mu n/h-\mu)}\Big) \frac{h^2}{\mu n^3}h_2^2} + \Exp{\frac{d}{2( \mu n/h-\mu)}\frac{h_2^2}{\mu n}}\nonumber\\
    &=\frac{h \left(d h+2 \mu ^2\right) (d n+2 h \mu )}{4 \mu ^2 n^3},
\end{align}
so that by strong duality, this is the optimal distribution for $h_j$. This yields for $c(h)$ that
\begin{align}\label{eq:ckcor}
    \max_{\pprob_1,\pprob_2\in\mathcal{P}(\mu,d)}\mathbb{E}_{\pprob_1,\pprob_2}[c(h)]&=\frac{n^2}{h^2}\frac{h \left(d h+2 \mu ^2\right) (d n+2 h \mu )}{4 \mu ^2 n^3}\nonumber\\
    & =\frac{ d^2}{4 \mu ^2 }(1+o(1)).
\end{align}

However, the maximal value of $c(h)$ is now attained by two different distributions for $h_i$ and $h_j$, while our objective was to maximize $c(h)$ with i.i.d. distribution for $h_i$ and $h_j$. We therefore now consider the uncorrelated three-point distribution $\pprob_3$ for $h_i$ and $h_j$ of
\begin{align}
    p_0& =\frac{d}{2\mu},\quad p_\mu=1-\frac{d}{2\mu}-\frac{d}{2(\sqrt{\mu n}-\mu)},\nonumber\\
    p_{\sqrt{\mu n}}& =\frac{d}{2 (\sqrt{\mu n}-\mu )},
\end{align}
which is a proper distribution as long as $d<2\mu(1-\sqrt{\mu/n})$.
This yields as expected value for $c(h)$ of
\begin{align}\label{eq:primalcase3}
    \mathbb{E}_{\pprob_3}[c(h)]& =\frac{n^2}{h^2}\frac{h^2}{(\mu n)^3}\mathbb{E}_{\pprob_3}[h^2]^2\nonumber\\
    & = \frac{1}{\mu^3 n}\Big(\frac{1}{2} d \sqrt{\mu  n}+\mu ^2\Big)^2\nonumber\\
    & = \frac{d^2}{4\mu^2}(1+o(1)),
\end{align}
which is the same leading order term as in~\eqref{eq:ckcor}, where correlations between $h_i$ and $h_j$ are allowed. Since
\begin{equation}
    \max_{\pprob_1,\pprob_2\in\mathcal{P}(\mu,d)}\mathbb{E}_{\pprob_1,\pprob_2}[c(h)]\geq \mathbb{E}_{\pprob_3}[c(h)],
\end{equation}
$\pprob_3$ is asymptotically optimal. Furthermore, an increase in $d$ does not affect the set of feasible dual solutions for the unconstrained problem. Thus,~\eqref{eq:dualoptcase3} is still an upper bound of the maximal $c(h)$. As by~\eqref{eq:ckcor} this dual value is asymptotically equal to $d^2/(4\mu^2)$, and~\eqref{eq:primalcase3} achieves the same value asymptotically, this must imply that $\pprob_3$ is asymptotically optimal when it is a proper probability distribution, thus for all $d<2\mu(1-\sqrt{\mu/n})$.
\end{proof}

\end{document}